\def\draft{1}  

\documentclass[11pt]{article}
\usepackage[letterpaper,hmargin=1in,vmargin=1.25in]{geometry}

\usepackage{geometry}                
\usepackage{fullpage}
\usepackage{graphicx}
\usepackage{amssymb}
\usepackage{epstopdf}
\usepackage{amsfonts,latexsym,graphics}
\usepackage{epsfig}
\usepackage{amssymb}
\usepackage{amsmath}
\usepackage{amsfonts}
\usepackage{algorithm}
\usepackage{algorithmic}
\usepackage{url}
\usepackage{rotating}
\usepackage{hyperref}
\usepackage{eepic}
\usepackage{sectsty}
\usepackage{enumitem}

\ifnum\draft=1
\newcommand{\Knote}[1]{{[\bf Kai-Min's Note: #1]}}
\newcommand{\Wnote}[1]{{[\bf Xiaodi's Note: #1]}}

\else
\newcommand{\Knote}[1]{{}}
\newcommand{\Wnote}[1]{{}}
\newcommand{\Wnote}[1]{{}}
\fi

\newcommand{\vx}{\vec{x}}

\newcommand{\va}{\vec{a}}
\newcommand{\val}{\mathrm{val}}
\newcommand{\eval}{\mathrm{val}^*}

\DeclareMathOperator*{\Ex}{\mathbb{E}}

\newcommand{\bures}{\mathcal{B}}

\newtheorem{theorem}{Theorem}

\newtheorem{lemma}[theorem]{Lemma}

\newtheorem{corollary}[theorem]{Corollary}
\newtheorem{claim}[theorem]{Claim}

\newtheorem{fact}[theorem]{Fact}

\newtheorem{remk}[theorem]{Remark}

\newenvironment{proof}{\noindent{\bf Proof. }}{\qed}

\def\FullBox{\hbox{\vrule width 8pt height 8pt depth 0pt}}

\def\qed{\ifmmode\qquad\FullBox\else{\unskip\nobreak\hfil
\penalty50\hskip1em\null\nobreak\hfil\FullBox
\parfillskip=0pt\finalhyphendemerits=0\endgraf}\fi}

\def\qedsketch{\ifmmode\Box\else{\unskip\nobreak\hfil
\penalty50\hskip1em\null\nobreak\hfil$\Box$
\parfillskip=0pt\finalhyphendemerits=0\endgraf}\fi}

\newcommand{\E}{\mathop{\mathrm E}\displaylimits}

\newcommand{\remove}[1]{}

\newcommand{\eps}{\varepsilon}

\newcommand{\bfa}{{\mathbf{a}}}
\newcommand{\bfb}{{\mathbf{b}}}
\newcommand{\bfx}{{\mathbf{x}}}
\newcommand{\bfy}{{\mathbf{y}}}

\def\01{\{0,1\}}
\def\eps{\epsilon}

\newcommand{\microspace}{\mspace{0.5mu}}

\newcommand{\tr}{\operatorname{tr}}

\newcommand{\ip}[2]{\left\langle #1 | #2\right\rangle}

\def\({\left(}
\def\){\right)}
\def\I{\mathsf{id}}

\def\<{\langle}
\def\>{\rangle}
\def \lket {\left|}
\def \rket {\right\rangle}
\def \lbra {\left\langle}
\def \rbra {\right|}
\newcommand{\ket}[1]{\lket\microspace #1 \microspace\rket}
\newcommand{\bra}[1]{\lbra\microspace #1 \microspace\rbra}
\newcommand{\ketbra}[1]{\lket\microspace #1 \rangle \langle #1 \microspace\rbra}

\def\X{\mathcal{X}}
\def\Y{\mathcal{Y}}

\def\A{\mathcal{A}}
\def\B{\mathcal{B}}

\def\E{\mathcal{E}}

\newcommand{\commentout}[1]{}
\numberwithin{theorem}{section}
\numberwithin{equation}{section}

\newenvironment{protocol*}[1]
  {
    \begin{center}
      \hrulefill\\
      \textbf{#1}
  }
  {
    \vspace{-1\baselineskip}
    \hrulefill
    \end{center}
  }

\newcommand{\Hmin}{H_{\mathrm{min}}}
\begin{document}

\title{Parallel repetition for entangled $k$-player games \\ via fast quantum search}
\author{Kai-Min Chung\\Academia Sinica\\ \texttt{kmchung@iis.sinica.edu.tw} \and Xiaodi Wu\\MIT\\\texttt{xiaodiwu@mit.edu} \and Henry Yuen\\MIT\\\texttt{hyuen@csail.mit.edu}}
\maketitle

\begin{abstract}
	We present two parallel repetition theorems for the entangled value of multi-player, one-round free games (games where the inputs come from a product distribution).  Our first theorem shows that for a $k$-player free game $G$ with entangled value $\mathrm{val}^*(G) = 1 - \epsilon$, the $n$-fold repetition of $G$ has entangled value $\mathrm{val}^*(G^{\otimes n})$ at most $(1 - \epsilon^{3/2})^{\Omega(n/sk^4)}$, where $s$ is the answer length of any player. In contrast, the best known parallel repetition theorem for the \emph{classical} value of two-player free games is $\mathrm{val}(G^{\otimes n}) \leq (1 - \epsilon^2)^{\Omega(n/s)}$, due to Barak, et al. (RANDOM 2009). This suggests the possibility of a separation between the behavior of entangled and classical free games under parallel repetition.

	Our second theorem handles the broader class of free games $G$ where the players can output (possibly entangled) quantum states. For such games, the repeated entangled value is upper bounded by $(1 - \epsilon^2)^{\Omega(n/sk^2)}$. We also show that the dependence of the exponent on $k$ is necessary: we exhibit a $k$-player free game $G$ and $n \geq 1$ such that $\mathrm{val}^*(G^{\otimes n}) \geq \mathrm{val}^*(G)^{n/k}$.
	
	Our analysis exploits the novel connection between communication protocols and quantum parallel repetition, first explored by Chailloux and Scarpa (ICALP 2014). We demonstrate that better communication protocols yield better parallel repetition theorems: in particular, our first theorem crucially uses a quantum search protocol by Aaronson and Ambainis, which gives a quadratic Grover speed-up for distributed search problems. Finally, our results apply to a broader class of games than were previously considered before; in particular, we obtain the first parallel repetition theorem for entangled games involving more than two players, and for games involving quantum outputs.
\end{abstract}

\section{Introduction}

The study of multi-player one-round games has been central to both theoretical computer science and quantum information. Games have served as an indispensible tool with which to study a diverse array of topics, from the hardness of approximation to cryptography; from delegated computation to Bell inequalities; from proof systems to the monogamy of entanglement. In particular, two-player games have received the most scrutiny. In a two-player game $G$, a referee samples a pair of questions $(x,y)$ from some distribution $\mu$, and sends question $x$ to one player (typically named Alice), and $y$ to the other (typically named Bob). Alice and Bob then utilize some non-communicating strategy to produce answers $a$ and $b$, respectively, upon which the referee computes some predicate $V(x,y,a,b)$ to decide whether to accept or not. In this paper, we focus on the setting where Alice and Bob may utilize quantum entanglement as part of their strategy. The primary quantity of interest is the \emph{entangled value} $\eval(G)$ of game $G$, which is the maximum success probability over all possible entangled strategies for the players.

Recently, there has been significant interest in the \emph{parallel repetition} of entangled games~\cite{kempe2011parallel, ChaillouxS14, chailloux2014parallel, JainPY14, DinurSV14}. More formally, the $n$-fold parallel repetition of a game $G$ is a game $G^{\otimes n}$ where the referee will sample $n$ independent pairs of questions $(x_1,y_1), \ldots, (x_n,y_n)$ from the distribution $\mu$. Alice receives $(x_1,\ldots,x_n)$ and Bob receives $(y_1,\ldots,y_n)$. They produce outputs $(a_1,\ldots,a_n)$ and $(b_1,\ldots,b_n)$, respectively, and they win only if $V(x_i,y_i,a_i,b_i) = 1$ for all $i$. We call each $i$ a ``coordinate'' of $G^{\otimes n}$ or ``repetition'' of $G$.

Suppose we have a game $G$ where $\eval(G) = 1 - \eps$. Intuitively, one should expect that $\eval(G^{\otimes n})$ should behave as $(1 - \eps)^n$. Indeed, this would be the case if the game $G$ were played $n$ times \emph{sequentially}. However, there are counterexamples of games $G$ and $n > 1$ where $\eval(G^{\otimes n}) = \eval(G)$ (see Section~\ref{sec:lower_bound}). Despite such counterexamples, it has been shown that the \emph{classical} value $\val(G^{\otimes n})$ (i.e. where the players are restricted to using classical strategies) of a repeated game $G^{\otimes n}$ goes down exponentially with $n$, for large enough $n$~\cite{raz1998parallel,holenstein2007parallel}. This result is known as the Parallel Repetition Theorem, and is central in the study of hardness of approximation, probabilistically checkable proofs, and hardness amplification in classical theoretical computer science. 

Recently, quantum analogues of the Parallel Repetition Theorem have been studied, and for certain types of games, it has been shown that the entangled game value also goes down exponentially with the number of repetitions. In particular, parallel repetition theorems have been shown for $2$-player \emph{free games} (see~\cite{ChaillouxS14, chailloux2014parallel, JainPY14}) and \emph{projection games} (see~\cite{DinurSV14}). Free games are where the input distribution to the players is a product distribution (i.e. each players' questions are chosen independently of each other). Projection games are where, for each answer of one designated player, there is at most one other answer for the other player that the referee would accept. 

Most relevant to this work are the results of~\cite{ChaillouxS14, chailloux2014parallel, JainPY14} on free entangled games. Among them, the best parallel repetition theorem was obtained by~\cite{chailloux2014parallel}, who prove that for a two-player free game $G$, the entangled value of the $n$-fold repetition is at most $(1 - \eps^2)^{\Omega(n/s)}$, where $s$ is the answer length of the players. When $G$ is also a projection game, they obtain \emph{strong parallel repetition}: the repeated game value is at most $(1 - \eps)^{\Omega(n)}$. The centerpiece of their analysis is a novel connection between communication complexity and parallel repetition of games. 

\subsection{Our results}
In this work, we further develop this connection between games and communication protocols to obtain improved parallel repetition theorems for free entangled games. We present a generic framework where one obtains parallel repetition theorems for free games by designing succinct communication protocols. The core concept we present is
\begin{center}
\textbf{\emph{Better parallel repetition theorems from better communication protocols.}}
\end{center}
The first instantiation of this concept is the following theorem:
\begin{theorem}
\label{thm:grover_pr_informal}
	Let $k \geq 2$ be an integer. Let $G$ be a $k$-player free game with entangled value $\eval(G) = 1-\eps$. Then, for $n = \Omega( sk^4 \log(k/\eps) /\eps^{3/2})$,
	$$
		\eval(G^{\otimes n}) \leq (1 - \eps^{3/2})^{\Omega(n/k^4 s)}
	$$
	where $s$ is the output answer length of the players.
\end{theorem}
The proof of Theorem~\ref{thm:grover_pr_informal} uses a quantum communication protocol that performs a version of distributed unstructured search (i.e. searching for a $1$ in a bitstring).   The improvement of the base from $1 - \eps^2$ (as found in~\cite{chailloux2014parallel}) to $1 - \eps^{3/2}$ comes from the fact that the unstructured search problem on $N$ bits can be solved by a quantum algorithm using only $O(\sqrt{N})$ queries. We discuss this in more detail in the next section. 

Our second theorem handles a broader class of games, where the players can output quantum states as answers. We are able to handle this broader class of games because our framework allows general quantum communication protocols.
\begin{theorem}
\label{thm:cq_pr_informal}
	Let $k \geq 2$ be an integer. Let $G$ be a $k$-player free game, where players can output (possibly entangled) quantum states, and has entangled value $\eval(G) = 1-\eps$. Then, for all $n$,
	$$
		\eval(G^{\otimes n}) \leq (1 - \eps^{2})^{\Omega(n/k^2 s)}
	$$
	where $s = \max_j \log (d_j)$, where $d_j$ is the dimension of player $j$'s output state.
\end{theorem}
Furthermore, we prove that the dependence of the exponent on the number of players $k$ is necessary:
\begin{theorem}
\label{thm:lowerbound_informal}
	For all $k \geq 2$, there exists a $k$-player free game $G$ and $n > 1$ where  $\eval(G^{\otimes n}) \geq \eval(G)^{n/k}$.
\end{theorem}
To our knowledge, our results are the first to show quantum parallel repetition in the setting of games with more than $2$ players.

Finally, we give a proof of parallel repetition for the \emph{classical} value of $k$-player free games. While this theorem appears to be a folklore result, we were not able to find any explicit proof of it. We provide one here for the sake of completeness.
\begin{theorem}
\label{thm:classical_pr_informal}
Let $G$ be a $k$-player free game with classical value $\val(G) = 1 - \eps$. Then 
$$
\val(G^{\otimes n}) \leq (1 - \eps^2)^{\Omega(n/sk)},
$$
where $s$ is the output answer length of the players.
\end{theorem}

\medskip
\noindent \textbf{CQ Games}. Our second theorem applies to a class of games that is a generalization of the traditional notion of games that involve two players and have classical inputs and outputs. In this paper we introduce the class of $k$-player \emph{classical-quantum} (CQ) games, where the players receive classical inputs, apply local unitary operators to their share of an entangled state, and return some qubits to the referee. The referee then makes a measurement on the answer qubits to decide whether to accept or reject. If we restrict the players' unitaries to be permutation matrices, and the referee's measurement to be diagonal in the standard basis, then we recover the class of classical games. 

We believe the model of CQ games is worth deeper investigation. One motivation for the study of CQ games comes from the recent exciting work of Fitzsimons and Vidick~\cite{fitzsimons2014multiprover}, who demonstrated an efficient reduction transforming a local Hamiltonian $H = H_1 + \cdots + H_m$ acting on $n$ qubits to a $5$-player CQ-game $G_H$ such that approximating $\eval(G_H)$ with inverse polynomial accuracy will decide whether the ground state energy of $H$ is a YES or NO instance of the QMA-complete problem \textsc{Local Hamiltonians}. In this game, the referee sends $O(\log n)$-sized questions, and the players responds with $O(1)$-qubit states as answers. The significance of this is that it opens up the possibility of proving a ``games'' version of the Quantum PCP conjecture. This intriguing possibility calls for further study of the behavior of CQ games. 

\subsection{Parallel repetition and communication protocols}
At a high level, most proofs of parallel repetition proceed via reduction. Let $G$ be a two-player free game with verification predicate $V(x,y,a,b)$. If there were a strategy $\mathcal{S}$ for the repeated game $G^{\otimes n}$ that wins with too large probability, then one can transform $\mathcal{S}$ to a strategy $\mathcal{T}$ to play a single instance of the game $G$ with probability larger than $\eval(G)$, which would be a contradiction. 

In~\cite{ChaillouxS14, chailloux2014parallel,JainPY14}, the reduction from a repeated game strategy to a single game strategy has two steps: (1) a ``too-good'' repeated game strategy $\mathcal{S}$ is converted to an \emph{advice-based strategy} for game $G$, which wins with high probability. An advice-based strategy is a collection of \emph{advice states} $\{\varphi_{xy}\}_{xy}$ so that when Alice and Bob receive inputs $x$ and $y$, they happen to share the entangled state $\varphi_{xy}$, which they can measure to produce answers. Of course, this is not a valid quantum strategy for game $G$, but (2) using the assumption that $\mathcal{S}$ has very high winning probability, the advice-based strategy can be \emph{rounded} to a true game strategy: Alice and Bob can apply local operations $U_x$ and $V_y$, respectively, on some input-independent state $\varphi$ to approximate $\varphi_{xy}$, and thus simulate the advice-based strategy (with some error). 

One can construct the advice states $\{\varphi_{xy}\}_{xy}$ from $\mathcal{S}$ in different ways. Generally, the goal is to create advice states that closely mimick the joint state of the players during an actual execution of the strategy $\mathcal{S}$, \emph{conditioned} on some event. Ideally, we would like to condition on the event that the players won all $n$ coordinates of $G^{\otimes n}$. This would give rise to the ideal advice states: whenever the players receive an input $x$ and $y$, their advice state $\varphi_{xy}$ would have precisely the correct answers $a$ and $b$ that would allow them to win the single game $G$.

However, it seems impossible to argue that such an ideal advice-based strategy can be simulated by a true game strategy. The approach taken by~\cite{ChaillouxS14, chailloux2014parallel} is to construct advice states $\{\varphi_{xy}\}$ that have two properties: (a) the advice-based strategy succeeds with high probability for $G$, and (b) there is a \emph{low-cost communication protocol} between Alice and Bob that produces $\varphi_{xy}$ when they receive inputs $x$ and $y$, respectively. Property (b) makes it possible to approximate the advice-based strategy using a valid quantum strategy: small communication complexity translates into small rounding error. 

The communication protocol used in~\cite{ChaillouxS14, chailloux2014parallel} is a simple one: Alice and Bob first play the optimal strategy $\mathcal{S}$ for $G^{\otimes n}$. They receive inputs $(x_1,\ldots,x_n)$ and $(y_1,\ldots,y_n)$, and measure a shared state $\ket{\xi}$ and obtain $n$-tuples of outputs $(a_1,\ldots,a_n)$ and $(b_1,\ldots,b_n)$. Then, Alice samples a small subset of coordinates $i_1,\ldots,i_h$, and sends over her inputs and outputs in this subset to Bob, who verifies that the original game $G$ was won in each of these coordinates. If Bob finds a $i_j$ such that the tuple $V(x_{i_j},y_{i_j},a_{i_j},b_{i_j}) = 0$, then Bob aborts. Otherwise, Bob accepts. If we condition the final state of the protocol on Bob accepting, then we have a very good proxy for the ideal advice states described above. But since the communication complexity of this protocol is small, we can round this to a valid quantum strategy. 

However, the analysis of~\cite{ChaillouxS14, chailloux2014parallel} is tailored to simple one-way communication protocols involving classical messages. We generalize this paradigm to show that, if the advice states $\{\varphi_{xy}\}$ can be constructed using \emph{any} communication protocol (which can be two-way, and involve quantum messages), then the advice-based strategy using $\{\varphi_{xy}\}$ can be simulated with a true game strategy, with error that is related to the communication complexity of the protocol. This unlocks a richer toolbox for the reduction designer: one can use many more tools from communication complexity to engineer good advice states. This gives rise to the concept of ``Better parallel repetition theorems from better communication protocols.''

Our theorems are instantiations of this mantra. At the heart of the communication protocol used in our first theorem is a variant of the Grover search algorithm. There, the players sample a random subset of coordinates $i_1,\ldots,i_h$ as before, but now they perform quantum search over the indices to find a ``losing coordinate'': i.e., a coordinate $i_j$ such that $V(x_{i_j},y_{i_j},a_{i_j},b_{i_j}) = 0$. The quadratic speedup of Grover's search algorithm translates into a quadratic savings in communication complexity, which is precisely what allows us to improve the base of the repeated game value from $1 - \eps^{2}$ to $1 - \eps^{3/2}$. For our second theorem, we take advantage of the fact that the communication protocol can be quantum, which allows us to handle games with quantum outputs.

Our use of quantum search in the protocol to generate the advice states gives a generic way to improve the reduction for arbitrary free games. However, one could also use this technique to prove \emph{game-specific} parallel repetition theorems. That is, one could try to leverage special properties of a particular game to design a succinct communication protocol for generating advice states, and in turn, obtain a parallel repetition theorem with better parameters. Indeed, one can see this idea in the result of~\cite{chailloux2014parallel} for free projection games: by using the projection property of the game, their communication protocol avoids sending whole input and output symbols. This allows them to prove a repeated game value of $(1 - \eps)^{\Omega(n)}$ -- note that this does not depend on the output alphabet!


\subsection{Related work}

We discuss how our result relates to prior results in parallel repetition, classical and quantum. Most relevant to our work are the results on free games. Jain, et al.~\cite{JainPY14} and Chailloux and Scarpa~\cite{ChaillouxS14,chailloux2014parallel} both proved that the entangled value of $2$-player free games (with classical inputs and outputs) goes down exponentially with the number of repetitions. In particular,~\cite{chailloux2014parallel} showed for such a game $G$ with $\eval(G) = 1-\eps$, we have that $\eval(G^{\otimes n}) \leq (1 - \eps^2)^{\Omega(n/s)}$, where $s$ is the output length of the players. They also show that, when $G$ is also a projection game, \emph{strong} parallel repetition holds: $\eval(G^{\otimes n}) \leq (1 - \eps)^{\Omega(n)}$. 


In a different line of work, Dinur, Steurer and Vidick show that projection games (with an arbitrary input distribution) also have an exponential decay in entangled value under parallel repetition: if $G$ be a $2$-player projection game with classical inputs and outputs, and $\eval(G) = 1-\eps$, then $\eval(G^{\otimes n}) \leq (1 - \eps^{12})^{\Omega(n)}$~\cite{DinurSV14}. This result is not comparable with our work, nor with the work of~\cite{chailloux2014parallel, JainPY14}. While~\cite{DinurSV14} can handle games with arbitrary input distributions, the games need to satisfy the projection property. On the other hand, the results on free games can handle arbitrary verification predicates, but the input distributions need to be product.

There is a rich history of study of parallel repetition in classical theoretical computer science, which we will not detail here. Most relevant to us is the work of Barak, et al.~\cite{barak2009strong}, who showed that for $2$-player free games $G$ with classical value $1 - \eps$, $\val(G^{\otimes n}) \leq (1 - \eps^2)^{\Omega(n/s)}$, where $s$ is the output length of the players. Intriguingly, it is not known whether the $\eps^2$ term is tight for free games (it is known that this is necessary for classical parallel repetition of general games~\cite{raz2011counterexample}). Our first theorem demonstrates a possible separation between classical and quantum parallel repetition; the base of our repeated game value is $1 - \eps^{3/2}$, rather than $1 - \eps^2$.

Finally, there has been little prior study of the parallel repetition of games with more than $2$ players. Buhrman, et. al. studied this question for non-signaling players, and showed that the non-signaling value of repeated games goes down exponentially with the number of repetitions~\cite{buhrman2013parallel}. Their parallel repetition theorem holds for games with \emph{full support}, meaning that every possible combination of questions gets asked with positive probability; furthermore, the rate of decay also depends on the complete description of the game, not just the original game value and the number of repetitions. Arnon-Friedman, et al. prove similar results for multi-player non-signaling games, but they use a new technique called de Finetti reductions~\cite{arnon2014non}. Rosen also studied $k$-player parallel repetition in a weaker version of the non-signaling model, and demonstrated an exponential rate of decay~\cite{Rosen10}. 

\section{Proof overviews}
\label{sec:outline}

\subsection{Overview of Theorem~\ref{thm:grover_pr_informal}}
Here we give a very informal outline of the proof of Theorem~\ref{thm:grover_pr_informal}, for the case of two-player free games. The full proof that handles an arbitrary number of players can be found in Section~\ref{sec:grover_pr} of the Appendix. 

Let $G$ be a two-player free game, with inputs $(x,y)$ drawn from a product distribution $\mu = \mu_X \otimes \mu_Y$, and with verification predicate $V(x,y,a,b)$. Let $\mathcal{S}$ be an optimal strategy for the repeated game $G^{\otimes n}$, where Alice and Bob shares an entangled state $\ket{\xi}$, and upon receiving a tuple of inputs $\bfx = (x_1,\ldots,x_n)$ and $\bfy = (y_1,\ldots,y_n)$, Alice and Bob perform local measurements $M^\bfx = \{M^\bfx_\bfa\}_\bfa$ and $N^\bfy = \{N^\bfy_\bfb \}_\bfb$ on their respective parts of $\ket{\xi}$, to obtain answer tuples $\bfa = (a_1,\ldots,a_n)$ and $\bfb = (b_1,\ldots,b_n)$. Let $\ket{\xi_{\bfx\bfy\bfa\bfb}}$ be the (unnormalized) post-measurement state of $\ket{\xi}$ after making measurements $(M^\bfx,N^\bfy)$, and obtaining outcomes $(\bfa,\bfb)$. We assume for contradiction that $\eval(G^{\otimes n}) > 2^{-\gamma n}$, for some small $\gamma$. 

\medskip
\noindent \textbf{A naive approach}. Consider the following state
$$
	\ket{\theta} = \frac{1}{\sqrt{\lambda}} \sum_{\bfx, \bfy} \sqrt{\mu^{\otimes n}(\bfx,\bfy)} \ket{\bfx} \otimes \ket{\bfy} \otimes  \sum_{\bfa,\bfb: V(\bfx,\bfy,\bfa,\bfb) = 1} \ket{\xi_{\bfx\bfy\bfa\bfb}} \otimes \ket{\bfa} \otimes \ket{\bfb}
$$
where $\lambda$ is a normalizing constant, $\mu^{\otimes n}$ is the input distribution for $G^{\otimes n}$, and $V(\bfx,\bfy,\bfa,\bfb) = \prod_i V(\bfx_i,\bfy_i,\bfa_i,\bfb_i)$. This would be the joint state of Alice and Bob if they received inputs $\bfx$ and $\bfy$ in coherent superposition, played strategy $\mathcal{S}$, and won the game. Here is a naive idea to use $\ket{\theta}$ in a strategy $\mathcal{T}$ for $G$: Alice and Bob share $\ket{\theta}$, with Alice possessing the $\ket{\bfx}$ input register, the $\ket{\bfa}$ output register, and half of $\ket{\xi_{\bfx\bfy\bfa\bfb}}$; Bob possesses the $\ket{\bfy}$ input register, the $\ket{\bfb}$ output register, and the other half of $\ket{\xi_{\bfx\bfy\bfa\bfb}}$. When they receive inputs $(x,y) \leftarrow \mu$, they each measure some fixed coordinate $i$ of their respective input registers of $\ket{\theta}$ to obtain input symbols $x'$ and $y'$, respectively. Suppose that $x = x'$ and $y = y'$: their shared state has collapsed to $\ket{\theta_{xy}}$. Then measuring the $i$th coordinate of their output registers will yield outputs $(a,b)$ such that $V(x,y,a,b) = 1$. Thus $\{ \theta_{xy} \}$ is an excellent set of advice states -- call this ensemble the \emph{ideal advice}. However, in general, this cannot be rounded to a valid game strategy. 

\medskip
\noindent \textbf{Advice states from Grover search}. Instead, we will construct another ensemble that mimicks the ideal advice, and if $\eval(G^{\otimes n})$ is too large, can be rounded to a valid game strategy with small error. We construct a state $\ket{\varphi}$ in steps. First, Alice and Bob start with
$$
	\ket{\psi^0} = \sum_{\bfx, \bfy} \sqrt{\mu^{\otimes n}(\bfx,\bfy)} \ket{\bfx} \otimes \ket{\bfy} \otimes  \sum_{\bfa,\bfb} \ket{\xi_{\bfx\bfy\bfa\bfb}} \otimes \ket{\bfa} \otimes \ket{\bfb}.
$$
Note that this state can be produced without any communication. Then, Alice and Bob engage in a short communication protocol to determine whether they've lost or won the game repeated: they need to determine whether there exists a coordinate $i$ such that $V(\bfx_i,\bfy_i,\bfa_i,\bfb_i) = 0$ -- call this a \emph{losing coordinate}. Classically, this would require $\Omega(n)$ bits of communication, which is too large for us. Instead, Alice and Bob can perform a distributed version of Grover's algorithm to search for a losing coordinate. Although Grover's search algorithm is a quantum query algorithm, it is a standard technique to convert query algorithms into communication protocols (see~\cite{buhrman1998quantum}): Alice executes the Grover search algorithm, and whenever she has to query the $i$th coordinate, Alice sends the query request to Bob, who responds with $(\bfy_i,\bfb_i)$. Alice can then compute $V(\bfx_i,\bfy_i,\bfa_i,\bfb_i)$. If Alice finds a losing coordinate, she aborts the protocol. Otherwise, she accepts. Since the Grover algorithm requires $O(\sqrt{n})$ queries, this communication protocol uses $\tilde{O}(\sqrt{n})$ qubits of communication, where $\tilde{O}(\cdot)$ hides the $\log n$ bits needed for the query request, as well as the input and output lengths. This protocol is performed coherently with the $\ket{\bfx}$, $\ket{\bfy}$, $\ket{\bfa}$, and $\ket{\bfb}$ registers. Let $\ket{\psi}$ denote the final state of this protocol, and let $\ket{\varphi}$ denote $\ket{\psi}$ \emph{conditioned} on Alice accepting.

If Grover search worked perfectly, then $\ket{\varphi}$ would be essentially the same as the naive $\ket{\theta}$ we described first. However, Grover's algorithm does not perform search perfectly, and has some error. Furthermore, when we condition on Alice not finding a losing coordinate, this error gets multiplied by $1/\eval(G^{\otimes n})$. Though we are assuming $\eval(G^{\otimes n})$ is ``large'', it is still exponentially small, and hence we require that the Grover search has exponentially small error. In our proof, we make some technical adjustments to the search protocol in order to handle this exponential blowup of the Grover error (without increasing the communication complexity to $\Omega(n)$ bits), but for the sake of exposition we will ignore this issue. For now, we can treat $\ket{\varphi}$ as a very good approximation of $\ket{\theta}$ -- thus, defining $\ket{\varphi_{xy}}$ in the same way we defined $\ket{\theta_{xy}}$ yields a good ensemble of advice states $\{ \varphi_{xy} \}$.

\medskip
\noindent \textbf{Rounding to a valid quantum strategy}. Now it remains to show that $\{ \varphi_{xy} \}$ can be rounded to a valid quantum strategy. We do this by establishing two properties of the state $\ket{\varphi}$: there exists a coordinate $i \in [n]$ such that
\begin{enumerate}
	\item Suppose we decohere (i.e. measure) the $i$th coordinate of the $\ket{\bfx\bfy}$ registers of $\ket{\varphi}$, and let $(X_i,Y_i)$ denote the random measurement outcomes. Then the distribution of $(X_i,Y_i)$ is $\gamma$-close to $\mu$, the input distribution of $G$; and
	\item Let $B$ denote the part of $\varphi$ controlled by Bob. Then the quantum mutual information between $X_i$ (after measurement) and $B$ in $\varphi$, denoted by $I(X_i : B)_\varphi$, is at most $\gamma$. Similarly, we have $I(Y_i : A)_\varphi \leq \gamma$, where $Y_i$ and $A$ are defined analogously.
\end{enumerate}

Property 1 follows from the fact that the distribution of $\bfx$ and $\bfy$, before conditioning, is a product distribution across coordinates $i$. When we condition on an event with probability $\lambda$, then on average, the distribution of the individual coordinates $(x_i,y_i)$ are skewed by at most $\sqrt{\log (1/\lambda)/n}$ in total variation distance. This simple but useful fact is known as \emph{Raz's Lemma} in the parallel repetition literature. Thus, if $\lambda \gg 2^{-n}$, then the input distribution of most coordinates, even after conditioning, is largely unaffected. Here, $\lambda$ corresponds to the probability that Alice does not abort the protocol, which is at least $\eval(G^{\otimes n})$.

Property 2 is the most interesting part of our proof. It states that Bob's part of the state $\varphi$ is relatively uncorrelated with the value of the input $X_i$, and similarly Alice's part of $\varphi$ is relatively uncorrelated with the value of $Y_i$. This uses the fact that our protocol has low communication complexity: intuitively, since Alice and Bob communicate at most $\tilde{O}(\sqrt{n})$ qubits in the protocol, they ``learn'' at most $\tilde{O}(\sqrt{n})$ bits total about each other's inputs. Amortized over the $n$ coordinates, this means Alice has about $1/\sqrt{n}$ bits of information about each $y_i$, on average, and similarly for Bob. When we condition on Alice not aborting, each player's knowledge of the other's inputs increases by at most $\log 1/\lambda$. If $\lambda > 2^{-\gamma n}$, Alice's state has $O(\gamma)$ mutual information with each $y_i$, on average. 

This intuition is formalized by leveraging the beautiful result of Nayak and Salzman that gives limits on the ability of entanglement-assisted quantum communication protocols to transmit classical messages~\cite{nayak2006limits}. More specifically, consider a general two-way quantum communication protocol between Alice and Bob, who may start with some shared entangled state. Suppose that Alice is given a uniformly random $m$-bit message $X$ at the beginning of the protocol. If $T$ qubits are exchanged between Alice and Bob over the course of the protocol, Bob can only guess Alice's input $X$ with probability at most $2^{2T}/2^m$. Equivalently, the mutual information between Bob's final state and $X$ is at most $2T$. Applying the Nayak-Salzman theorem to our setting, and using what we call \emph{Quantum Raz's Lemma}\footnote{We make a quick remark about Quantum Raz's Lemma. The ingredients of Quantum Raz's Lemma can be found in various forms in~\cite{ChaillouxS14, chailloux2014parallel, JainPY14}, but we find it conceptually advantageous to consolidate these ingredients into a single Lemma that is used as a black box. The benefit of this consolidation is that the overarching structure of the proofs of Theorems~\ref{thm:grover_pr_informal} and \ref{thm:cq_pr_informal} are the same -- really, the only essential difference is the communication protocol!}, we can conclude that on average, $I(Y_i : A)_\varphi = I(X_i : B)_\varphi = \frac{1}{n} \left( \tilde{O}(\sqrt{n}) + \log 1/\lambda \right) = O(\gamma)$.


Once we have established Property 1 and 2, then the \emph{Quantum Strategy Rounding Lemma} (which can be found in both~\cite{ChaillouxS14, JainPY14}) then gives that there exists a unitaries $\{U_x\}_x$ and $\{V_y\}_y$ for Alice and Bob, respectively, so that $U_x\otimes V_y \ket{\varphi} \approx \ket{\varphi_{xy}}$. Thus we have a valid quantum strategy for $G$: on input $(x,y)$, Alice and Bob locally apply unitaries $U_x$ and $V_y$ to their shared state $\varphi$, and obtain something close to the advice state $\varphi_{xy}$, which they can use to win game $G$ with probability close to $1$. For sufficiently small $\gamma$, this will be greater than $\eval(G)$, a contradiction. Thus, $\eval(G^{\otimes n}) \leq 2^{-\gamma n}$. This concludes the proof outline. 

\medskip
\noindent \textbf{Other technical considerations}. While this discussion has been very informal, it captures the conceptual arguments that are required by our analysis. There are many technical details that are handled by the full proof in the Appendix: for example, in order to make the error in the Grover search exponentially small, we increase the communication complexity of the protocol to $\tilde{O}(\log (1/\lambda)/\sqrt{\eps})$. If $\lambda < 2^{-\eps^{3/2} n}$, where $\eval(G) = 1- \eps$, then the communication complexity is at most $\tilde{O}(\eps n)$, which is still small enough for use in Quantum Raz's Lemma. Another issue is that the communication protocol described requires that Bob transmit his input symbols $y_i$, which would incur a dependence on the input alphabet size. Through a modification of the protocol and the analysis, we are able to avoid this dependence. Finally, instead of using Grover's algorithm exactly, we use the 3-dimensional search algorithm of Aaronson and Ambainis~\cite{aaronson2003quantum}, which performs quantum search in a ``spatially local'' way. When converted to a communication protocol, the parties no longer need to incur a $\log n$-qubit overhead per round simply to transmit a query request.

The arguments above are not specific to two-players. We prove our theorem for the general $k$-player case. An important part of this is the $k$-player generalization of the Quantum Strategy Rounding lemma of~\cite{ChaillouxS14,JainPY14}, which we prove in Lemma~\ref{lem:superjpy}. 

\subsection{Overview of Theorem~\ref{thm:cq_pr_informal}}

Theorem~\ref{thm:cq_pr_informal} shows parallel repetition for the entangled value of $k$-player free CQ games: these are games where the players may produce quantum states as answers (see Section~\ref{sec:cq-proof} of the Appendix for a formal definition of the model). Instead of making measurements on their share of the entangled state, players will apply local unitaries (that depend on their inputs), and transmit a number of qubits to the referee. The referee will then perform some joint verification measurement on all the answer qubits to decide whether to accept or not.

Similarly to the proof of Theorem~\ref{thm:grover_pr_informal}, we will use a low-cost communication protocol to design advice states for the repeated-game-to-single-game reduction. However, we were not able to use the distributed Grover search technique. Instead, our low-cost communication protocol performs the following (in the two-player setting): Alice will send Bob her inputs and answer qubits corresponding to coordinates in a small subset $C \subseteq [n]$. Bob will then perform the referee's verification measurement on Alice's inputs and outputs, and his own inputs and outputs, to determine whether they won all the coordinates in the subset $C$. Having determined whether they won or not, Bob will return Alice's message back to her. If $C$ is sufficiently small, then the communication cost of this task is small.

This simple checking protocol is similar to the checking protocol of~\cite{chailloux2014parallel}. However, in our protocol, Alice and Bob exchange quantum messages, and it is a two-way protocol (because Bob has to return Alice's message back to her). The theorem of~\cite{nayak2006limits} again allows us to show that $I(X_i : B)_\varphi$ and $I(Y_i : A)_\varphi$ are small, where $\varphi$ is the advice state that arises from this communication protocol, and thus we can apply quantum strategy rounding as before.

\medskip
\noindent \textbf{Outline.} In Section~\ref{sec:prelim}, we list the quantum information theoretic facts we'll need, as well as prove a few useful technical lemmas (including Quantum Raz's Lemma). In Section~\ref{sec:strategy_rounding}, we prove our $k$-player Quantum Strategy Rounding lemma. We prove Theorem~\ref{thm:grover_pr_informal} in Section~\ref{sec:grover_pr}, and Theorem~\ref{thm:cq_pr_informal} in Section~\ref{sec:cq-proof}. Our lower bound example demonstrating the necessary dependence on the number of players $k$ is in Section~\ref{sec:lower_bound}, and the proof of $k$-player parallel repetition for free classical games can be found in the Appendix.

\section{Preliminaries}
\label{sec:prelim}

We assume familiarity with the basics of quantum information and computation. For a comprehensive reference, we refer the reader to~\cite{nielsen2010quantum,wilde2013quantum}. For a pure state $\ket{\psi}$, we will let $\psi$ denote the density matrix $\ketbra{\psi}$. If $\ket{\psi}^{AB}$ is a bipartite state, then $\psi^A$ will be the reduced density matrix of $\psi^{AB}$ on space $A$. A density matrix $\rho^{XA}$ is a \emph{classical-quantum} (CQ) state if $\rho^{XA} = \sum_x p(x) \ketbra{x} \otimes \rho^A_x$, where $p(x)$ is a probability distribution and $\rho^A_x$ is an arbitrary density matrix on space $A$. For a probability distribution $\mu$, $x \leftarrow \mu$ indicates $x$ is drawn from $\mu$. For a classical state $\rho^X = \sum_x \mu(x) \ketbra{x}$, we write $x \leftarrow \rho^X$ to denote $x \leftarrow \mu$. We let $\I$ denote the identity matrix.

\subsection{$k$-player games}
We give a formal definition of games, where the inputs and outputs are classical. In Section~\ref{sec:cq-proof} we will give a more general definition of \emph{CQ games}, where the outputs may be quantum. A $k$-player game is a tuple $G = (\X,\A,\mu,V)$, where:
\begin{enumerate}
\item $\X = \X_1 \times \cdots \times \X_k$ with each $\X_j$ a finite alphabet,
\item $\A = \A_1 \times \cdots \A_k$ with each $\A_j$ a finite alphabet,
\item $\mu$ is a distribution over $\X$,
\item $V: \X \times \A \to \{0,1\}$ is the verification predicate.
\end{enumerate}
In a $k$-player $G$, a referee samples an input $x = (x_1,\ldots,x_k)$ from $\mu$, and sends $x_j$ to player $j$. The players produce a vector of outputs $a = (a_1,\ldots,a_k)$ (where the $j$th player outputs symbol $a_j$, and the referee accepts if $V(x,a) = 1$.

We say a game is \emph{free} if $\mu = \mu_1 \otimes \cdots \otimes \mu_k$, where $\mu_j$ is a distribution over $\X_j$ (i.e. $\mu$ is a product distribution). A \emph{quantum strategy} for $G$ is a shared state $\ket{\xi}^E$ (where $E$ are $k$-partite spaces split between the $k$ players), and for each player $j$ a set of measurements $\{M^{j,x_j}\}_{x_j \in \X_j}$ (with each $M^{j,x_j}$ being a set of POVM elements $\{M^{j,x_j}_{a_j}\}_{a_j \in \A_j}$) which act on the space $E_j$. On input $x_j$, player $j$ measures the $E_j$ register of $\ket{\xi}$ using measurement $M^{j,x_j}$, and obtains an outcome $a_j$, which is then sent to the referee. The entangled value of a game $G$ is defined as the maximum probability a referee will accept over all possible (finite-dimensional) quantum strategies for $k$ players:
$$
	\eval (G) = \max_{ \ket{\xi}, \{\{M^{j,x_j}\}_{x_j} \}_j} \Ex_{x \leftarrow \mu} \left[ \sum_{a \in \A: V(x,a) = 1} \tr \left( M^{1,x_1}_{a_1} \otimes \cdots \otimes M^{k,x_k}_{a_k} \xi \right) \right].
$$

The $n$-fold repetition of a game $G = (\X,\A,\mu,V)$ is denoted by $G^{\otimes n} = (\X^n, \A^n, \mu^{\otimes n}, V^n )$, where: $\mu^{\otimes n}$ is the product distribution over $n$ independent copies of $\X$, and $V^n(\vx,\va) := \prod_i V(\vx_i,\va_i)$, with $\vx \in \X^n$ and $\va \in \A^n$.

\subsection{Properties of the squared Bures metric}

For two positive semidefinite operators $\rho, \sigma$, let the fidelity between $\rho$ and $\sigma$ be denoted by $F(\rho,\sigma) := \tr \sqrt{\rho^{1/2} \sigma \rho^{1/2}}$. The fidelity distance measure has the well-known property that for pure states $\ket{\psi}$ and $\ket{\varphi}$, $F(\psi,\varphi) = |\ip{\psi}{\varphi}|$. Furthermore, when $\rho$ and $\sigma$ are classical probability distributions in the same basis (i.e. $\rho = \sum_i p_i \ketbra{i}$ and $\sigma = \sum_i q_i \ketbra{i}$), then $F(\rho,\sigma) = \sum_i \sqrt{p_i q_i}$.

The fidelity distance measure is not a metric on the space of positive semidefinite operators. For one, it does not satisfy a triangle inequality. However, one can convert fidelity into other measures that are metrics. One such measure is the \emph{Bures metric}, defined as $\bures(\rho,\sigma) := \sqrt{1 - F(\rho,\sigma)}$. In this paper, we will use the \emph{squared Bures metric}, denoted by $K(\rho,\sigma) := \bures(\rho,\sigma)^2$, as the primary distance measure between quantum states. It satisfies many pleasant properties, including the following:

\begin{fact}[Triangle inequality]
\label{fact:fidelity_triangle_inequality}
	Let $n \geq 2$ and let $\rho_1,\ldots,\rho_{n+1}$ be density matrices. Then
	$$
		K(\rho_1,\rho_{n+1}) \leq n \sum_i K(\rho_i,\rho_{i+1}).
	$$
\end{fact}
\begin{proof}
	We adapt the proof from~\cite{chailloux2014parallel}. For $i \in [n]$ let $\alpha_i = \arccos(F(\rho_i,\rho_{i+1}))$. Let $\alpha = \arccos(F(\rho_1,\rho_{n+1}))$. Then, since $\arccos(F(\cdot,\cdot))$ is a distance measure for quantum states, we have $\alpha \leq \sum_i \alpha_i$. Then we have
	$$
		K(\rho_1,\rho_{n+1}) = 1 - \cos(\alpha) \leq n^2(1 - \cos(\alpha/n)) \leq n\sum_i (1 - \cos(\alpha_i)) = n \sum_{i=1}^n K(\rho_i,\rho_{i+1}).
	$$
\end{proof}

\begin{fact}[Contractivity under quantum operations]
\label{fact:fidelity_contractivity}
	Let $\E$ be a quantum operation, and let $\rho$ and $\sigma$ be density matrices. Then $K(\E(\rho),\E(\sigma)) \leq K(\rho,\sigma)$.
\end{fact}

\begin{fact}[Unitary invariance]
\label{fact:fidelity_unitary_invariance}
	Let $U$ be unitary, and let $\rho$ and $\sigma$ be density matrices.Then $K(U\rho U^\dagger, U\sigma U^\dagger) = K(\rho,\sigma)$.
\end{fact}

\begin{fact}[Convexity]
\label{fact:concavity_of_fidelity}
	Let $\{A_i\}$ and $\{B_i\}$ be finite collections of positive semidefinite operators, and let $\{p_i\}$ be a probability distribution. Then $K(\sum_i p_i A_i \sum_i p_i B_i) \leq \sum_i p_i K(A_i,B_i)$.
\end{fact}

\begin{fact}
\label{fact:fidelity_between_cq}
	Let $\{A_i\}$ and $\{B_i\}$ be finite collections of positive semidefinite operators, and let $\{p_i\}$ be a probability distribution. Then $K(\sum_i p_i \ketbra{i} \otimes A_i, \sum_i p_i \ketbra{i} \otimes B_i) = \sum_i p_i K(A_i,B_i)$.
\end{fact}

\subsection{Quantum information theory}
For two positive semidefinite operators $\rho$, $\sigma$, the relative entropy $S(\rho \| \sigma)$ is defined to be $\tr(\rho (\log \rho - \log \sigma))$. The relative min-entropy $S_\infty(\rho \| \sigma)$ is defined as $\min\{ \lambda : \rho \preceq 2^\lambda \sigma \}$. The entropy of $\rho$ is denoted by $H(\rho) := -\tr(\rho \log \rho)$. For a tripartite state $\rho^{ABC}$, the conditional mutual information $H(A | B)_\rho$ is defined as $H(\rho^{AB}) - H(\rho^{B})$. Let $\rho^{AB}$ be a bipartite state. Then the mutual information $I(A:B)_\rho$ is defined as $H(A)_\rho - H(A | B)_\rho$. An equivalent definition is $I(A:B)_\rho = S(\rho^{AB} \| \rho^A \otimes \rho^B)$. 

\begin{fact}[\cite{jain2003lower}]
\label{fact:mutual_information_vs_fidelity}
	Let $\rho$ and $\sigma$ be density matrices. Then $S(\rho \| \sigma) \geq K(\rho,\sigma)$.
\end{fact}

\begin{fact}[\cite{jain2003lower}]
\label{fact:avg_divergence}
	Let $\mu$ be a probability distribution on $\X$. Let $\rho = \sum_{x \in X} \mu_x \ketbra{x} \otimes \rho^A_x$. Then $I(X : A)_{\rho} = \Ex_{x \leftarrow \mu}[S(\rho_x \| \rho)]$.
\end{fact}

\begin{fact}[\cite{JainPY14}, Fact II.11]
\label{fact:divergence_contractivity}
Let $\rho^{XY}$ and $\sigma^{XY}$ be quantum states. Then $S(\rho^{XY} \| \sigma^{XY} ) \geq S(\rho^{X} \| \sigma^X)$.
\end{fact}

\begin{fact}
\label{fact:divergence_split_rule}
Let $\rho^{XY}$ and $\sigma^{XY} = \sigma^{X} \otimes \sigma^Y$ be quantum states. Then $S(\rho^{XY} \| \sigma^{XY} ) \geq S(\rho^{X} \| \sigma^X) + S(\rho^{Y} \| \sigma^Y)$.
\end{fact}

\begin{fact}[\cite{JainPY14}, Fact II.8]
\label{fact:divergence_chain_rule}
	Let $\rho = \sum_x \mu(x) \ketbra{x} \otimes \rho_x$, and $\rho^1 = \sum_x \mu^1(x) \ketbra{x} \otimes \rho^1_x$. Then $S(\rho^1 \| \rho) = S(\mu^1 \| \mu) + \Ex_{x \leftarrow \mu^1} \left [ S(\rho^1_x \| \rho_x) \right]$.
\end{fact}


\begin{fact}[\cite{JainPY14}, Lemma II.13]
\label{fact:max_divergence}
	Let $\rho = p \rho_0 + (1 - p) \rho_1$. Then $S_\infty(\rho_0 \big \| \rho) \leq \log 1/p$.
\end{fact}

\begin{fact}
\label{fact:relative_min_entropy_contractivity}
	Let $\rho^{AB}$ and $\sigma^{AB}$ be density matrices. Then $S_\infty(\rho^{AB} \| \sigma^{AB}) \geq S_\infty(\rho^A \| \sigma^B)$. 
\end{fact}

\begin{fact}
\label{fact:relative_min_entropy_chain_rule}
Let $\rho$, $\sigma$, and $\tau$ be density matrices such that $S_\infty(\rho \| \sigma) \leq \lambda_1$ and $S_\infty(\sigma \| \tau) \leq \lambda_2$. Then $S_\infty(\rho \| \tau) \leq \lambda_1 + \lambda_2$.
\end{fact}

\begin{fact}
\label{fact:relative_min_entropy_chain_rule2}
Let $\rho$, $\sigma$, and $\tau$ be density matrices such that $S(\rho \| \sigma) \leq \lambda_1$ and $S_\infty(\sigma \| \tau) \leq \lambda_2$. Then $S_\infty(\rho \| \tau) \leq \lambda_1 + \lambda_2$.
\end{fact}
\begin{proof}
	$S_\infty(\sigma \| \tau) = \lambda_2$ implies that $2^{-\lambda_2} \sigma \preceq \tau$. Then,
	\begin{align*}
		S(\rho \| \tau) &= \tr(\rho (\log \rho - \log \tau)) \\
				   &\leq \tr(\rho(\log \rho - \log 2^{-\lambda_2} \sigma)) \\
				   &\leq \tr(\rho(\log \rho - (-\lambda_2)\I - \log \sigma)) \\
				   &\leq \lambda_2 + \tr(\rho (\log \rho - \log \sigma)) \\
				   &= \lambda_1 + \lambda_2.
	\end{align*}
\end{proof}

\subsection{Some technical lemmas}
The following lemma is due to~\cite{barak2009strong}:
\begin{lemma}[\cite{barak2009strong}, Lemma 3.3]
\label{lem:brrrs-divergence}
	Let $P = (p,1-p)$ and $Q = (q,1-q)$ be binary distributions. If $S(P \| Q) \leq \delta$, and $p < \delta$, then $q \leq 4\delta$. 
\end{lemma}
The following adapts Lemma~\ref{lem:brrrs-divergence} to use the distance measure $K$ instead:
\begin{lemma}
\label{lem:brrrs-fidelity}
	Let $P = (p,1-p)$ and $Q = (q,1-q)$ be binary distributions. If $K(P,Q) \leq \delta$, and $p < \delta$, then $q \leq 4\delta$. 
\end{lemma}
\begin{proof}
	If $q \leq p$, then we are done. Assume otherwise. We have that $\delta \geq K(P,Q) = 1 - F(P,Q) \geq (1 - F(P,Q)^2)/2$, because $0 \leq F(P,Q) \leq 1$. $F(P,Q)^2 = (\sqrt{pq} + \sqrt{(1 - p)(1 - q)})^2 = pq + 1 - p - q + pq + 2\sqrt{pq(1 - p)(1 - q)}$, and thus
	\begin{align*}
		2\delta &\geq p + q -2pq - 2\sqrt{pq(1 - p)(1 - q)} \\
			   &\geq p + q - 2pq - 2\sqrt{pq} \\
			   &= (\sqrt{p} - \sqrt{q})^2 - 2pq \\
			   &\geq (\sqrt{p} - \sqrt{q})^2 - 2\delta,
	\end{align*}
where in the last line we used the assumption that $p \leq \delta$. Then $2\sqrt{\delta} \geq | \sqrt{p} - \sqrt{q} | \geq \sqrt{q}$, and thus $q \leq 4\delta$.
\end{proof}

Finally, we prove a quantum analogue of Raz's Lemma, which is the central tool behind many information-theoretic proofs of parallel repetition theorems~\cite{raz1998parallel, holenstein2007parallel,barak2009strong}: 
\begin{lemma}[Quantum Raz's Lemma]
\label{lem:quantum_raz}
	Let $\psi^{XA} = \left( \sum_{x} \mu(x) \ketbra{x}^X\right) \otimes \psi^A$ be a CQ-state, classical on $X$ and quantum on $A$, where $X$ is $n$-partite. Furthermore, suppose that $\mu(x) = \prod \mu_i(x_i)$. Let $\varphi^{XA} = \sum_{x} \sigma(x) \ketbra{x}^X \otimes \varphi^A_x$ be such that $S(\varphi \| \psi) \leq t$. Then, 
	$$
		\sum_i I(X_i : A)_\varphi \leq 2t.
	$$
\end{lemma}
\begin{proof}
	First observe the following manipulations:
	\begin{align*}
		t &\geq S(\varphi^{XA} \| \psi^{XA}) \\
		  &= S(\varphi^{XA} \| \psi^X \otimes \psi^A) \\
		  &\geq S(\varphi^{XA} \| \varphi^X \otimes \varphi^A) \\
		  &= I(X : A)_\varphi \\
		  &= H(X)_\varphi - H(X | A)_\varphi \\
		  &\geq H(X)_\varphi - \sum_i H(X_i | A)_\varphi.
	\end{align*}
	We focus on $H(X)_\varphi$ now. Using that relative entropy is always non-negative:
	\begin{align*}
		-H(X)_\varphi + \sum_i H(X_i)_\varphi &\leq -H(X)_\varphi + \sum_i S(\varphi^{X_i} \| \psi^{X_i}) + H(X_i)_\varphi \\
									&= -H(X)_\varphi - \sum_i \tr( \varphi^{X_i} \log \psi^{X_i}) \\
									&= -H(X)_\varphi - \tr(\varphi^{X} \log \psi^{X}) \\
									&= S(\varphi^{X} \| \psi^{X}) \\
									&\leq t.
	\end{align*}
	Continuing, we have
	$$
		t \geq -t + \sum_i H(X_i)_\varphi - H(X_i | A)_\varphi = -t + \sum_i I(X_i : A)_\varphi.
	$$
\end{proof}

\section{Quantum strategy rounding}
\label{sec:strategy_rounding}

In this section we prove our $k$-player Quantum Strategy Rounding lemma, generalizing the technique of~\cite{chailloux2014parallel,JainPY14}. 

\begin{lemma}[\cite{JainPY14}]
\label{lem:easyjpy}
	Let $\mu$ be a probability distribution on $\X$. Let
	$$
		\ket{\varphi} := \sum_{x \in \X} \sqrt{\mu(x)} \ket{xx}^{X X'} \otimes \ket{\varphi_x}^{AB}.
	$$
	Let $\ket{\varphi_x} := \ket{xx}^{XX'} \otimes \ket{\varphi_x}^{AB}$. Then there exists unitary operators $\{ U_x \}_{x \in \X}$ acting on $XX' A$ such that
	$$
		\Ex_{x \leftarrow \mu} \left[ K(\varphi_x, U_x \varphi\, U_x^\dagger) \right] \leq I(X : B)_{\varphi}.
	$$
\end{lemma}

\begin{proof}
	We follow the proof in~\cite{JainPY14}. Denote the reduced states of Bob by $\rho_x := \tr_{XX'A}(\varphi_x)$ and $\rho := \tr_{X X' A} (\varphi)$. By Facts~\ref{fact:mutual_information_vs_fidelity} and~\ref{fact:avg_divergence}, we get that
	$$
		I(X : B)_{\varphi} = \Ex_{x \leftarrow \mu}[S(\rho_x \| \rho)] \geq \Ex_{x \leftarrow \mu}[K(\rho_x , \rho)].
	$$
	By Uhlmann's Theorem, for each $x \in \X$ there exists $U_x$ such that $| \bra{\varphi_x} (U_x \otimes \I_B) \ket{\varphi} | = F(\rho_x,\rho)$. Furthermore, this is equal to $F(\varphi_x, U_x \otimes \I_B\, \varphi\, U_x^\dagger \otimes \I_B)$. We thus obtain the claim.
\end{proof}

\begin{lemma}
\label{lem:fidelity_between_distributions}
	Let $\{ \ket{\varphi_a} \}_{a \in \A}$ be a finite collection of pure states. Let $\mu$ and $\tau$ be probability distributions over $\A$ such that $S(\mu \| \tau) \leq \eps$. Then 
	$$
		K( \Ex_{a \leftarrow \mu} [\ketbra{\varphi_a}],\Ex_{a \leftarrow \tau} [\ketbra{\varphi_a}]) \leq \eps.
	$$
\end{lemma}
\begin{proof}
	Consider the states $\ket{\psi^{\mu}} = \sum_{a \in \A} \sqrt{\mu_a} \ket{aa}^{AA'} \otimes \ket{\varphi_a}$ and $\ket{\psi^{\tau}} = \sum_{a \in \A} \sqrt{\tau_a} \ket{aa}^{AA'} \otimes \ket{\varphi_a}$. Let $\rho^\mu = \tr_{A'} (\ketbra{\psi^\mu})$ and $\rho^\tau = \tr_{A'} (\ketbra{\psi^\tau})$. Then notice that $\Ex_{a \leftarrow \mu} [\ketbra{\varphi_a}] = \tr_{AA'}(\rho^\mu)$ and $\Ex_{a \leftarrow \tau} [\ketbra{\varphi_a}] = \tr_{AA'}(\rho^\tau)$, respectively. We then have that, considering the partial trace as a quantum operation, $K( \Ex_{a \leftarrow \mu} [\ketbra{\varphi_a}],\Ex_{a \leftarrow \tau} [\ketbra{\varphi_a}]) \leq K(\rho^\mu,\rho^\tau)$. By Uhlmann's Theorem, this is at most $1 - |\ip{\psi^\mu}{\psi^\tau}| = 1 - \sum_{a \in \A} \sqrt{\mu_a \tau_a} = K(\mu,\tau)$. By Fact~\ref{fact:mutual_information_vs_fidelity}, this is at most $S(\mu \| \tau) \leq \eps$.
\end{proof}

\begin{lemma}[Quantum strategy rounding]
\label{lem:superjpy}
	Let $k \geq 1$. Let $\mu$ be a probability distribution over $\X = \X_1 \times \X_2 \times \cdots \times \X_k$, where the $\X_i$ are finite alphabets. Let 
	$$
		\ket{\varphi} := \sum_{x \in \X} \sqrt{\mu(x)} \ket{xx}^{XX'} \otimes \ket{\varphi_{x}}^{AB}
	$$
	where $X = X_1\cdots X_k$, $X' = X_1'\cdots X_k'$, and $A = A_1\cdots A_k$ are $k$-partite registers. 
	Then for all $i \in [k]$ there exist operators $\{U^i_a\}_{a \in \X_i}$ acting on $X_i X_i' A_i$ such that
	$$
		\Ex_{x \leftarrow \mu} \left[ K \left(   \varphi_{x} , (U^1_{x_1} \otimes \cdots \otimes U^k_{x_k}) \, \varphi \, (U_{x_1}^{1,\dagger} \otimes \cdots \otimes U_{x_k}^{k,\dagger}) \right) \right] \leq 4k \sum_i I(X_i : X_{-i} X_{-i}' A_{-i} B)_\varphi,
	$$	
	where $A_{-i}$, $X_{-i}$, and $X_{-i}'$ denote the $A$, $X$, and $X'$ registers excluding the $i$th coordinate, respectively, and for all $x \in \X$, $\ket{\varphi_x} := \ket{xx} \otimes \ket{\varphi_x}$.  
\end{lemma}
\begin{proof}
	For $i \in [k]$, let $\nu_i = \mu_1\otimes \cdots \otimes \mu_i \otimes \mu_{> i}$, where $\mu_j$ denotes the marginal distribution of $\mu$ on coordinate $j$, and $\mu_{> i}$ denotes the marginal distribution of $\mu$ on coordinates $i +1 ,\ldots, k$. For $x \in \X$, for all $S \subseteq [k]$, let $x_S$ denote the coordinates of $x$ that are in $S$. Therefore, $x_{\leq i} = x_{1 \ldots i}$, and $x_{> i} = x_{i+1 \ldots k}$, etc.
	 For all $i \in [k]$ and $x\in \X$, define
$$
	\ket{\varphi_{x_{> i}}} := \ket{x_{>i}x_{>i}}^{X_{> i}X_{>i}'} \otimes \left ( \sum_{x_{\leq i}} \sqrt{\mu(x_{\leq i}| x_{> i})} \ket{x_{\leq i}x_{\leq i}}^{X_{\leq i}X_{\leq i}'} \otimes \ket{\varphi_x}^{AB} \right)
$$
and
	$$
	\ket{\varphi_{x_{i}}} := \ket{x_{i} x_{i}}^{X_{i} X_{i}'} \otimes \left(\sum_{x_{-i}} \sqrt{\mu(x_{- i} | x_{i})} \ket{x_{- i} x_{- i}}^{X_{-i} X_{- i}'} \otimes \ket{\varphi_x}^{AB} \right).
$$
Note that for all $i$, $\ket{\varphi} = \sum_{x_i} \sqrt{\mu_{i}(x_i)}\ket{\varphi_{x_{i}}} $. Then by Lemma~\ref{lem:easyjpy}, we get that there exists unitaries $\{U^i_{u}\}_{u \in \X_i}$ acting on $X_i X_i' A_i$ such that
$$
	\Ex_{x_i \leftarrow \mu_i} [K(\varphi_{x_i}, \mathcal{U}^i_{x_i}(\varphi))] \leq I(X_i : X_{-i} X_{-i}' A_{-i} B)_\varphi,
$$
where $\mathcal{U}^i_{x_i}$ is the CP map $\sigma \mapsto U^i_{x_i} \sigma (U^i_{x_i})^\dagger$. Define $\ket{\tilde{\varphi}_{x_{> i}}} = \ket{x_{>i}}^{X_{>i}''} \otimes \ket{\varphi_{x_{>i}}}$, $\ket{\tilde{\varphi}_{x_{i}}} = \ket{x_{i}}^{X_{i}''} \otimes \ket{\varphi_{x_i}}$, and $\ket{\tilde{\varphi}_x} = \ket{x}^{X''} \otimes \ket{\varphi_x}$. For notational convenience, let $\eps_i = I(X_i : X_{-i} X_{-i}' A_{-i} B)_\varphi$, and let $x$, $x_i$ and $x_{> i}$ denote the pure states $\ketbra{x}$, $\ketbra{x_i}$, and $\ketbra{x_{ > i}}$ respectively. 

Define the following states: $\rho_0 = \Ex_{x \leftarrow \mu} [ x \otimes \varphi_x ] = \Ex_{x \leftarrow \mu} [ \tilde{\varphi}_x ]$, and for all $i \in [k]$, $\rho_i = \Ex_{x \leftarrow \nu_i} [ x \otimes \mathcal{U}_{x_{\leq i}} (\varphi_{x_{> i}}) ]$, where $\mathcal{U}_{x_{\leq i}}$ denotes the CP map $\sigma \mapsto \left (\bigotimes_{j\leq i} U^j_{x_j} \right) \sigma \left (\bigotimes_{j\leq i} U^j_{x_j} \right)^\dagger$.
Then by the triangle inequality for the squared Bures metric (Fact~\ref{fact:fidelity_triangle_inequality}),
$$
	K(\rho_0,\rho_n) \leq k \sum_{i=0}^{k-1} K(\rho_i,\rho_{i+1}).
$$
We upper bound each term $K(\rho_i,\rho_{i+1})$:
\begin{align*}
	&K \left ( \Ex_{x \leftarrow \nu_i} [ x \otimes \mathcal{U}_{x_{\leq i}} (\varphi_{x_{> i}}) ], \Ex_{x \leftarrow \nu_{i+1}} [ x \otimes \mathcal{U}_{x_{\leq i+1}} (\varphi_{x_{> i+1}}) ] \right) \\
	&\leq \Ex_{x_{\leq i} \leftarrow \mu_1 \otimes \cdots \otimes \mu_i} K \left( \mathcal{U}_{x _{\leq i}} \left( \Ex_{x_{> i} \leftarrow \mu_{> i}} [x_{> i} \otimes \varphi_{x _{> i}}]\right), \mathcal{U}_{x_{\leq i}}\left( \Ex_{x_{> i} \leftarrow \mu_{i+1} \otimes \mu_{> i+1}} [x_{> i} \otimes  \mathcal{U}^{i+1}_{x_{i+1}} (\varphi_{x_{> i+1}}) ]\right) \right) \\
	&= K \left(\Ex_{x_{> i} \leftarrow \mu_{> i}} [x_{> i} \otimes  \varphi_{x_{> i}}], \Ex_{x_{> i} \leftarrow \mu_{i+1} \otimes \mu_{> i+1}} [x_{> i} \otimes \mathcal{U}^{i+1}_{x_{i+1}} (\varphi_{x_{> i+1}})]\right) \\
	&= K \left(\Ex_{x_{> i} \leftarrow \mu_{> i}} [\tilde{\varphi}_{x_{> i}}], \Ex_{x_{i+1} \leftarrow \mu_{i+1}} \left [x_{i+1} \otimes \mathcal{U}^{i+1}_{x_{i+1}} \left ( \Ex_{x_{> i+1} \leftarrow \mu_{> i+1}} [ \tilde{\varphi}_{x_{> i+1}}] \right) \right ] \right)
	\end{align*}
The second and third lines follow from the convexity and unitary invariance of the squared Bures metric, respectively (Facts~\ref{fact:concavity_of_fidelity} and Fact~\ref{fact:fidelity_unitary_invariance}). Consider the operation $\E$ that measures the registers $X_{> i+1} = X_{i+2} \ldots X_{k}$ in the standard basis, and copies the outcomes into new registers $X_{> i+1}''$. Then $\E(\Ex_{x_{i+1} \leftarrow \mu_{i+1}}  [\tilde{\varphi}_{x_{i+1}}]) =  \Ex_{x_{> i} \leftarrow \mu_{> i}} [\tilde{\varphi}_{x _{> i}}]$ and $\E(\varphi) =  \Ex_{x_{> i+1} \leftarrow \mu_{> i+1}} [ \tilde{\varphi}_{x_{> i+1}}] $. Then since $\E$ commutes with $\mathcal{U}^{i+1}_{x_{i+1}}$ and doesn't act on the $X_{i+1}''$ register, we have that the line above is equal to
\begin{align*}
&=K \left(\E \left( \Ex_{x_{i+1} \leftarrow \mu_{i+1}}  [\tilde{\varphi}_{x_{i+1}}] \right), \E \left( \Ex_{x_{i+1}\leftarrow \mu_{i+1}} [x_{i+1} \otimes \mathcal{U}^{i+1}_{x_{i+1}} (\varphi)] \right) \right) \\
&\leq K \left(\Ex_{x_{i+1} \leftarrow \mu_{i+1}}  [\tilde{\varphi}_{x_{i+1}}], \Ex_{x_{i+1}\leftarrow \mu_{i+1}} [x_{i+1} \otimes \mathcal{U}^{i+1}_{x_{i+1}} (\varphi)]\right) \\
&= \Ex_{x_{i+1} \leftarrow \mu_{i+1}}  K \left( \varphi_{x_{i+1}}, \mathcal{U}^{i+1}_{x_{i+1}} (\varphi) \right) \\
&\leq \eps_{i+1}.
\end{align*}
To complete the proof, we use the triangle inequality once more:
\begin{align*}
	\Ex_{x \leftarrow \mu} K(\varphi_x,\mathcal{U}_x(\varphi)) &= K \left( \Ex_{x \leftarrow \mu} [ \tilde{\varphi}_{x} ] , \Ex_{x \leftarrow \mu} [x \otimes \mathcal{U}_x (\varphi )]\right) \\
	&\leq 2K \left( \Ex_{x \leftarrow \mu} [ \tilde{\varphi}_{x} ] , \Ex_{x \leftarrow \nu_k} [x \otimes \mathcal{U}_x (\varphi) ]\right) + 2K \left( \Ex_{x \leftarrow \nu_k} [x \otimes \mathcal{U}_x (\varphi)], \Ex_{x \leftarrow \mu} [x\otimes \mathcal{U}_x( \varphi)]\right) \\
	&\leq 2k \sum_i \eps_i + 2k \sum_i \eps_i \\
	&\leq 4k \sum_i \eps_i .
\end{align*}
where $\mathcal{U}_x$ is the composition of $\mathcal{U}^i_{x_i}$ for all $i \in [k]$. Here we used Lemma~\ref{lem:fidelity_between_distributions} in the second line, and the fact that $S(\mu \| \nu_k) = I(X_1 : X_2 : \cdots : X_k)_\mu$, which is the multipartite mutual information between the coordinates of $X$. It is a known fact (see, e.g.,~\cite{yang2009squashed}) that the multipartite mutual information can be written in terms of the (standard) bipartite mutual information like so:
$$
	I(X_1 : X_2 : \cdots : X_k)_\mu \leq I(X_1: X_2)_\mu + I(X_1 X_2 : X_3)_\mu + \cdots + I(X_1 X_2 \cdots X_{k-1} : X_k)_\mu,
$$
but by the data processing inequality, we have that for all $i$, $I(X_1 \cdots X_{i-1} : X_i)_\mu \leq I(X_{-i} : X_i)_\mu \leq I(X_i : X_{-i} X_{-i}' A_{-i} B)_\varphi = \eps_i$. 
\end{proof}

\section{Parallel repetition using fast quantum search}
\label{sec:grover_pr}

\textbf{Notation}. Let $G = (\X,\A,\mu,V)$ be a $k$-player free game. In what follows, we will think of $x \in \X^n$ as $n\times k$ matrices, where the $i$th row indicates the inputs of all $k$ players in the $i$th coordinate, and the $j$th column indicates the inputs of the $j$th player. Thus $x(i,\cdot)$ indicates the $i$th row of $x$, and $x(\cdot,j)$ indicates the $j$th column. When we write $x_S$ for some subset $S \subseteq [n]$, we mean the submatrix of $x$ consisting of the rows indexed by $i \in S$. 

Let $X$ be an $n \times k$-partite register. Then we will also format $X$ as a $n \times k$ matrix, so $X_{(i,\cdot)}$ and $X_{(\cdot,j)}$ have the natural meaning. For a subset $S \subseteq [n]$, $X_S$ denotes the registers corresponding to the rows of $X$ indexed by $S$. For an index $j$, $X_{(S,j)}$ denotes the $j$th column of the rows indexed by $S$. $X_{(S,-j)}$ denotes the submatrix of $X$ corresponding to rows indexed by $i \in S$, and all columns except for the $j$th one. 

We make the following observation, which will be useful for us in our analysis: without loss of generality, we can restrict our attention to free games whose input distribution is the uniform distribution over some alphabet. Let $G = (\X,\A,\mu,V)$ be a $k$-player free game. Write $\mu = \mu_1 \otimes \cdots \otimes \mu_k$, where $\mu_j$ is a distribution over the alphabet $\X_j$. Fix an $\gamma > 0$. For each $i$, there exists an alphabet $\X_j'$ and a map $f_j: \X_j' \to \X_j$ such that the random variable $X_j' = f_j(U_j)$ (where $U_j$ is a uniformly random element from $\X_j')$ is $\gamma/k$-close in total variation distance to being distributed according to $\mu_j$ -- and hence the distribution of $(f_1(U_1),\ldots,f_k(U_k))$ is at most $\gamma$-far from $\mu$. Thus, we can ``simulate'' the game $G$ with another game $G' = (\X',\A,U,V')$, where $\X' = \X'_1 \times \cdots \X'_k$, $U$ is the uniform distribution on $\X'$, and $V': \X' \times \A \to \{0,1\}$ is the map $(x',a) \to V( \langle f_1(x'_1), \ldots,f_k(x'_k) \rangle, a)$.

\begin{claim}
\label{clm:uniform_simulation}
	$\eval(G') = \eval(G) \pm \gamma$.
\end{claim}
\begin{proof}
	Consider the optimal strategy for $G$. Then a strategy for $G'$ is the following: player $j$, on input $u'_j \in \X'_j$, computes $u_j = f_j(u'_j)$, and performs the strategy she would've done for $G$. The input distribution, from the point of view of the strategy for $G$, is at most $\gamma$-far from the original input distribution $\mu$. Thus the winning probability is at least $\eval(G) - \gamma$.
	
	Now consider the optimal strategy for $G'$. Then a strategy for $G$ is the following: player $j$, on input $u_j \in \X_j$, computes a uniformly random preimage $u'_j \in f_j^{-1}(u_j)$, and performs the strategy she would've done for $G'$. The input distribution, from the point of view of the strategy for $G'$, is at most $\gamma$-far from the uniform distribution $U$. Thus the winning probability is at least $\eval(G') - \gamma$.
\end{proof}

Furthermore, this simulation ``commutes'' with parallel repetition, in that $\eval((G')^{\otimes n}) = \eval(G^{\otimes n}) \pm \gamma n$. We can make $\gamma$ arbitrarily small, at the cost of (potentially) increasing the input alphabet size, so that the behavior of the simulation $G'$ is essentially the same as the original game $G$. However, since our theorems do not depend on the input alphabet size, we will treat $\gamma$ as infinitesimally small, and hence neglect it.

\begin{theorem}
\label{thm:grover_pr}
	Let $G = (\X,\A,\mu,V)$ be a $k$-player free game with classical outputs and classical verification predicate $V: \X \times \A \to \{0,1\}$. Suppose that $\eval(G) = 1 - \eps$. Let $s = \max_j \log |\A_j|$. Then, for all $n > k^4 s \log(k^2/\eps)/\eps^{3/2}$, we have that
	$$
		\eval(G^{\otimes n}) \leq (1 - \eps^{3/2})^{\Omega(n/k^4 s)}.
	$$
\end{theorem}
\begin{proof}
Because of Claim~\ref{clm:uniform_simulation}, it is without loss of generality to assume that the input distribution $\mu$ is the uniform distribution -- the following analysis can be performed on a simulation of $G$, which will still bound the repeated game value of $G$. 

Let $n$ be some integer greater than $k^4 s \log(k^2/\eps)/\eps^{3/2}$, and consider an optimal entangled strategy for $G^{\otimes n}$, and let $2^{-t}$ denote its winning probability. Suppose for contradiction that $t \leq c \epsilon^{3/2} n/ (k^4 s)$ for some universal constant $c$. Using this strategy, we will construct the following state
$$
	\ket{\varphi}^{XX'EA} := \sum_{x \in \X^n} \sqrt{\nu(x)} \ket{xx}^{XX'} \otimes \ket{\varphi_{x}}^{EA},
$$
where $\nu(x)$ is a probability distribution over $\X^n$, $X$, $X'$, $A$ are $n\times k$-partite registers, and $E$ is a $k$-partite register. We will show that exists a coordinate $i\in [n]$, and $\delta < \eps/32k^2$ satisfying the following properties:
\begin{enumerate}[label=(\Alph*)]
	\item Measuring the $X_{(i,\cdot)} A_{(i,\cdot)}$ register of $\varphi$ yields a tuple $(x_{(i,\cdot)},a_{(i,\cdot)})$ that satisfies $V(x_{(i,\cdot)},a_{(i,\cdot)}) = 1$ with probability at least $1 - \eps/8$;
	\item $S(\varphi^{X_{(i,\cdot)}} \| \mu) \leq \delta$.
	\item For all $j \in [k]$, $ I(X_{(i,j)} : Z_{-j})_{\varphi}  \leq \delta$, where $Z_{-j} = X'_{(\cdot,-j)} X_{(\cdot,-j)} E_{-j} A_{(\cdot,-j)}$.
\end{enumerate}
For now, we assume the existence of such a state $\ket{\varphi}$; we will construct it in Lemma~\ref{lem:state_construction}. We use Lemma~\ref{lem:superjpy} on the state $\varphi$ to obtain for each player $j$ a set of unitaries $\{U^j_{u} \}_{u \in \X_j}$ acting on $X_{(\cdot,j)} X'_{(\cdot,j)} E_j A_{(\cdot,j)}$ such that
\begin{align*}
		\Ex_{x_{(i,\cdot)} \leftarrow \varphi^{X_{(i,\cdot)}}}  \left[ K \left(  \varphi_{x_{(i,\cdot)}} ,\mathcal{U}_{x_{(i,\cdot)}} (\varphi) \right) \right] \leq 4 k \sum_j I(X_{(i,j)} : Z_{-j})_{\varphi}  \leq 4k^2 \delta,
\end{align*}
where we let $U_{x_{(i,\cdot)}} =  \bigotimes_j U^j_{x_{(i,j)}} $, and let $\mathcal{U}_{x_{(i,\cdot)}}$ be the CP map $\varphi \mapsto U_{x_{(i,\cdot)}} \varphi U_{x_{(i,\cdot)}}^\dagger$. The state $ \varphi_{x_{(i,\cdot)}}$ denotes $\varphi$ conditioned on $X_{(i,\cdot)} = x_{(i,\cdot)}$.

We now describe a protocol for the $k$ players to play game $G$. The players receive $u \in \X$, drawn from the product distribution $\mu$. Player $j$ receives $u_j \in \X_j$. The players share the state $\varphi$, where player $j$ has access to the $X_{(\cdot,j)}X_{(\cdot,j)}' E_j A_{(\cdot,j)}$ registers. 

\begin{figure}[H]
\begin{center}
\textbf{Protocol A} \\
\medskip
\framebox{
\begin{minipage}{0.9\textwidth}
	\textbf{Input}: $u \in \X$. Player $j$ receives $u_j$. \\
	\textbf{Preshared entanglement}: $\varphi$ \\
\textbf{Strategy for player $j$}:
\begin{enumerate}
	\item Apply the local unitary $U^j_{u_j}$ on the $X_{(\cdot,j)} X'_{(\cdot,j)} E_j A_{(\cdot,j)}$ registers of $\varphi$.
	\item Output the $A_{(i,j)}$ part of $\varphi$. 
\end{enumerate}

\end{minipage}
}

\end{center}
\end{figure}

Slightly overloading notation, we let $V^i_u$ denote the projector $\sum_{a \in \A : V(u,a) = 1} \ketbra{a}$ that acts on the $A_{(i,\cdot)}$ registers. Let $\kappa$ denote the winning probability of Protocol A. This is equal to
\begin{align*}
	\kappa &=  \Ex_{u \leftarrow \mu} \left \| V_u^i\,  U_{u} \, \ket{\varphi} \right \|^2\\
		 &\geq \Ex_{u \leftarrow \varphi^{X_{(i,\cdot)}}} \left \| V_u^i\,  U_{u} \, \ket{\varphi} \right \|^2 - 4\delta.
\end{align*}
where we use property \textbf{(B)} and appeal to Lemma~\ref{lem:brrrs-divergence}. Let $$\tau := \Ex_{u \leftarrow \varphi^{X_{(i,\cdot)}}} \left \| V_u^i\,  U_{u} \, \ket{\varphi} \right \|^2.$$
For every $i \in [n]$, $u \in \X$, define the quantum operation $\E_{i,u}$ that, given a state $\varphi$, measures the $A_{(i,\cdot)}$ registers using $V^i_{u}$ measurement, and outputs a classical binary random variable $F$ indicating the verification measurement outcome (outcome $1$ corresponds to ``accept'' and outcome $0$ corresponds to ``reject''). Let 
$$F_0 =  \Ex_{x_{(i,\cdot)} \leftarrow \varphi^{X_{(i,\cdot)}}} \E_{i,x_{(i,\cdot)}} \left (  \varphi_{x_{(i,\cdot)}} \right )\qquad \text{and} \qquad F_1 = 
\Ex_{u \leftarrow \varphi^{X_{(i,\cdot)}}} \E_{i,u} \left (  \mathcal{U}_{u} ( \varphi) \right).$$
Note that $\Pr(F_0 = 1) \geq 1 - \eps/8$ by our assumption on $\varphi$, and $\Pr(F_1 = 1) = \tau$. Then, 
\begin{align*}
	K(F_0,F_1) &= K \left( \Ex_{x_{(i,\cdot)} \leftarrow \varphi^{X_{(i,\cdot)}}} \E_{i,x_{(i,\cdot)}} \left (  \varphi_{x_{(i,\cdot)}} \right ), \Ex_{u \leftarrow \varphi^{X_{(i,\cdot)}}} \E_{i,u} \left (  \mathcal{U}_{u} ( \varphi) \right)\right) \\
	&\leq \Ex_{x_{(i,\cdot)} \leftarrow \varphi^{X_{(i,\cdot)}}}  K \left( \E_{i,x_{(i,\cdot)}} \left (  \varphi_{x_{(i,\cdot)}} \right ), \E_{i,x_{(i,\cdot)}} \left (  \mathcal{U}_{x_{(i,\cdot)}} ( \varphi) \right)\right) \qquad &\text{(Fact~\ref{fact:concavity_of_fidelity})} \\
	&\leq \Ex_{x_{(i,\cdot)} \leftarrow \varphi^{X_{(i,\cdot)}}}  K \left( \varphi_{x_{(i,\cdot)}} ,  \mathcal{U}_{x_{(i,\cdot)}} ( \varphi) \right) \qquad &\text{(Fact~\ref{fact:fidelity_contractivity})} \\
	&\leq 4k^2 \delta.
\end{align*}

By our assumption on $\delta$, this is at most $K(F_0,F_1) \leq \eps/8$. By Lemma~\ref{lem:brrrs-fidelity}, $\Pr(F_1 = 1) \geq 1 - \eps/8 - \eps/2$. Thus $\kappa \geq 1 - 3\eps/4$. But notice that Protocol A is a valid strategy for the game $G$; thus we have produced a strategy for game $G$ that wins with probability strictly greater than $1 - \eps$, a contradiction. Thus, it must be at $t = \Omega(\epsilon^{3/2} n/ (k^4 s))$, which establishes the theorem.
\end{proof}

\subsection{Construction of $\varphi$}

\begin{lemma}
\label{lem:state_construction}
There exists a state $\ket{\varphi}$, and a coordinate $i \in [n]$ satisfying properties \textbf{(A)}, \textbf{(B)}, and \textbf{(C)}.
\end{lemma}
\begin{proof}
Set $\eps' = \eps/32$, $\eta = 2^{-t} \eps/32k^2$, and $h = c' \log (1/\eta)/\eps'$ for some constant $c'$. We have $h =  (32c'/\eps)(t + \log (32k^2/\eps))$, and by our assumptions on $t$ and $n$, this is at most $n/2$.

Suppose there was a strategy to win the repeated game $G^{\otimes n}$ with probability $2^{-t}$, involving a shared state $\ket{\xi}^{E}$ (where $E$ is a $k$-partite state register) and measurements $\{M^{j,x_j}_a\}$ for the players, respectively. That is, player $j$, on input $x_j \in \X_j^{n}$, applies the measurement with POVM elements $\{M^{j,x_j}_a\}$ and reports the outcome. 

%
%

We will build the state $\varphi$ in steps. Consider the initial state
$$
	\ket{\psi^0} := \sum_{x\in \X^n} \sqrt{\mu^{\otimes n}(x)} \ket{xx}^{XX'} \otimes \sum_{a \in \A^n} \ket{\xi_{xa}}^{E} \otimes \ket{a}^A
$$
where $\ket{\xi_{xa}} = \left( \bigotimes_j \sqrt{M^{j,x_j}_{a_j}} \right) \ket{\xi}$ (which is a subnormalized state), and $\mu^{\otimes n}(x)$ is the probability distribution associated with the repeated game $G^{\otimes n}$. For every set $C \subset [n]$, and every fixing of the inputs $x_C$ to the coordinates indexed by $C$, define the state $\ket{\psi^0_{C,x_C}}$ to be $\ket{\psi^0}$ conditioned on $X_C = x_C$.



Now consider the following $k$-player communication protocol: for every set $C \subset [n]$ and every $x_C$, the players share the entangled state $\ket{\psi^0_{C,x_C}}$, where player $j$ has access to the registers $X_{(\cdot,j)} X'_{(\cdot,j)} E_j A_{(\cdot,j)}$. Using shared randomness, the players sample $h$ independent and uniformly random coordinates $C = \{i_1,\ldots,i_h \} \subset [n]$, and sample $x_C$ from the marginal distribution of $\mu^{\otimes n}$ on the subset $C$.  For the remainder of the protocol, the players perform all their operations on the shared state $\ket{\psi^0_{C,x_C}}$.

In the next phase of the protocol, the $k$ players communicate qubits to each other to determine whether they have won or lost the parallel repeated game $G^{\otimes n}$. In particular, they run a protocol to search for a coordinate $i \in C$ such that $V(x_{(i,\cdot)},a_{(i,\cdot)}) = 0$, if it exists -- call such a coordinate a losing coordinate. The state $\ket{\psi^0_{C,x_C}}$ becomes transformed to
$$
	\ket{\psi_{C,x_C}}^{XX'EAR} := \sum_{x\in \X^n} \sqrt{\mu^{\otimes n}(x|x_C)} \ket{xx}^{XX'} \otimes  \sum_{a \in \A^n}  \ket{\xi'_{C xa}}^{E} \otimes \ket{a}^{A} \otimes (\alpha_{C xa} \ket{1} + \beta_{C xa} \ket{0})^R
$$
where $\mu^{\otimes n}(x|x_C)$ is probability of $x$ conditioned on $x_C$, and $\ket{\xi'_{C xa}} = \ket{\xi_{xa}} \otimes \ket{w_{C xa}}$ with $\ket{w_{C xa}}$ denoting the workspace qubits that are used during the protocol. The coefficients $\alpha_{C xa}$ and $\beta_{C xa}$ denote the amplitude that the search protocol places on the flags ``No losing coordinates'' and ``Exists a losing coordinate'' respectively.

For now, we will abstract away from the particulars of this communication protocol and defer the details of it until later. The only things we will use about this search protocol is the following:
\begin{enumerate}
	\item The search protocol is run conditioned on $C$, and the $XA$ registers;
	\item At most $T = O(\sqrt{1/\eps'} \log (1/\eta) \log |\A|)$ qubits in total are exchanged between all parties, where $\A$ is the output alphabet in game $G$;
	\item For every fixing of $(x,a) \in \X^n \times \A^n$, if there are no coordinates $i \in [n]$ such that $V(x_{(i,\cdot)},a_{(i,\cdot)}) = 0$, then the search procedure reports ``No losing coordinates'' with certainty; and
	\item If there are at least an $\epsilon' n$ bad coordinates, then the search procedure reports ``No losing coordinates'' with probability at most $\eta$ (over the quantum randomness of the protocol, as well as over the choice of $C$). In other words, for tuples $(x,a)\in \X^n \times \A^n$ with $\Ex_i [V(x_{(i,\cdot)},a_{(i,\cdot)})] < 1 - \eps'$, 
$$
	\sum_C p(C) \, |\alpha_{C xa}|^2  \leq \eta,
$$
where $p(C)$ is the distribution that samples $h$ independent and uniformly random coordinates from $[n]$.
\end{enumerate}
For all $C$, $x_C$ define $\ket{\varphi_{C,x_C}}$ to be $\ket{\psi_{C,x_C}}$ conditioned on measuring $1$ in the $R$ register:
$$
	\ket{\varphi_{C,x_C}}^{XX'EAR} := \frac{1}{\sqrt{\lambda_{C,x_C}}} \sum_{x\in \X^n} \sqrt{\mu^{\otimes n}(x|x_C)} \ket{xx}^{XX'} \otimes  \sum_{a \in \A^n}  \ket{\xi'_{C xa}}^{E} \otimes \ket{a}^{A} \otimes (\alpha_{C xa} \ket{1}^R) 
$$
where $\lambda_{C,x_C}$ is for normalization. In the case that $\lambda_{C,x_C} = 0$ (meaning that we were trying to normalize the $0$ state), we leave the state undefined. Let $\psi^{CX_C}(C,x_C) = p(C) \mu^C(x_C)$ denote the joint probability distribution of the shared random variables $C$ and $X_C$, before conditioning. Let $\varphi^{CX_C}(C,x_C) = p(C) \mu^C(x_C) \lambda_{C,x_C}/\lambda$ denote the joint distribution conditioned on $R = 1$, where $\lambda = \sum_{C,x_C} \psi^{CX_C}(C,x_C) \lambda_{C,x_C}$.

\medskip
\noindent \textbf{(A). A random coordinate of $\varphi$ wins with high probability}. Let $$
	\rho = \Ex_{C,x_C \leftarrow \psi^{CX_C}} \left[ \ketbra{C} \otimes \ketbra{x_C} \otimes  \psi_{C,x_C} \right]\qquad \text{and} \qquad \sigma = \Ex_{C,x_C \leftarrow \varphi^{CX_C}} \left[ \ketbra{C} \otimes \ketbra{x_C} \otimes  \varphi_{C,x_C} \right].
$$
Observe that $\sigma$ is the post-measurement state of $\rho$ after measuring $\ket{1}$ in the $R$ register. Let $\E$ denote the quantum operation on that, (1) measures the $C$ register, (2) chooses a uniformly random $i \notin C$, (3) measures  $X_{(i,\cdot)}$ register, and (4) then conditioned on $X_{(i,\cdot)} = x_{(i,\cdot)}$, performs the binary verification measurement $V^i_{x_{(i,\cdot)}}$ defined in the previous section, setting an auxiliary register $Q$ to $\ket{1}$ if the measurement accepts, $\ket{0}$ if it rejects. We wish to argue that the probability that a measurement of the $Q$ register of $\E(\sigma)$ yields $1$ with high probability. This probability is equivalent to the probability the following process succeeds: first, measure the $XA$ registers of $\sigma$ to obtain a tuple $(x,a)$. Then, measure the $C$ register. Finally, select a random index $i \notin C$, and we succeed if $V(x_{(i,\cdot)},a_{(i,\cdot)}) = 1$.

In this alternative process, the probability that we measure $(x,a)$ in $\sigma$ such that $\Ex_{i \in [n]} [V(x_{(i,\cdot)},a_{(i,\cdot)})] < 1 - \eps'$ (call such $(x,a)$'s ``bad'') is equal to
$$
	\frac{1}{\lambda} \sum_{\text{$(x,a)$ bad}} \Pr_\rho (x,a) \sum_C p(C) |\alpha_{Cxa}|^2
$$
where $\Pr_\rho(x,a)$ is the probability of measuring measuring $(x,a)$ in $\rho$. By our assumption on the communication protocol, this is at most $\eta/\lambda$. Since the players' strategy wins the repeated game $G^{\otimes n}$ with probability $2^{-t}$, we have that $\lambda \geq 2^{-t}$. Thus the probability of measuring a bad $(x,a)$ is at most $2^t \eta$.

Now suppose we measure $(x,a)$ such that $\Ex_{i \in [n]} [V(x_{(i,\cdot)},a_{(i,\cdot)})] \geq 1 - \eps'$. Then, for any $C$, a random $i \notin C$ loses with probability at most $\eps' n/(n - |C|) \leq \eps' n/(n - h) \leq \eps/16$. Thus, the probability that the $Q$ register of $\E(\sigma)$ yields $0$ is at most $2^t \eta + \eps/16$. 

\medskip
\noindent \textbf{(B). Coordinate input distributions are mostly unaffected}. By Fact~\ref{fact:max_divergence}, since $\sigma \preceq 2^{\lambda} \rho$, we have
\begin{align}
	\log 1/\lambda &\geq S_\infty(\sigma \| \rho) \nonumber \\
				&\geq S(\sigma \| \rho) \nonumber \\
				&\geq  \Ex_{C,x_C \leftarrow \varphi^{CX_C}} S(\varphi_{C,x_C}^{XX'EA} \| \psi_{C,x_C}^{XX'EA}) \label{eq:grover_max_div},
\end{align}
where in the last line we used Fact~\ref{fact:divergence_chain_rule}. Using Facts~\ref{fact:divergence_contractivity} and~\ref{fact:divergence_split_rule}, we obtain that 

\begin{align*}
\log 1/\lambda &\geq \Ex_{C,x_C \leftarrow \varphi^{CX_C}} S(\varphi_{C,x_C}^{X} \| \psi_{C,x_C}^{X}) \\
			&\geq \Ex_{C,x_C \leftarrow \varphi^{CX_C}} \sum_{i \notin C} S(\varphi_{C,x_C}^{X_{(i,\cdot)}} \| \psi_{C,x_C}^{X_{(i,\cdot)}}) \\
			&= \Ex_{C,x_C \leftarrow \varphi^{CX_C}} \sum_{i \notin C} S(\varphi_{C,x_C}^{X_{(i,\cdot)}} \| \mu).
\end{align*}

\medskip
\noindent \textbf{(C). Mutual information is small}.  
\begin{claim}
\label{clm:relative_min_entropy}
Fix a $j\in [k]$, and fix a $C, x_C$. There exists a state $\sigma_{C,x_C}^{Z_{-j}}$ such that
$$
	S_\infty(\psi_{C,x_C}^{X_{(\cdot,j)} Z_{-j}} \| \psi_{C,x_C}^{X_{(\cdot,j)}} \otimes \sigma_{C,x_C}^{Z_{-j}}) \leq 2T,
$$
where $Z_{-j} = X_{(\cdot,-j)}X'_{(\cdot,-j)}  E_{-j} A_{(\cdot,-j)}$.
\end{claim}
We defer the proof of this claim for later, and will assume it for now. Line~\eqref{eq:grover_max_div} with Fact~\ref{fact:divergence_contractivity} implies that for all $j$, $\Ex_{C,x_C \leftarrow \varphi^{CX_C}} S(\varphi_{C,x_C}^{X_{(\cdot,j)} Z_{-j}} \| \psi_{C,x_C}^{X_{(\cdot,j)} Z_{-j}}) \leq \log 1/\lambda$. Using Fact~\ref{fact:relative_min_entropy_chain_rule2} with Claim~\ref{clm:relative_min_entropy}, we get that for all $j$, there exists a $\sigma_{C,x_C}^{Z_{-j}}$ such that
$$\Ex_{C,x_C \leftarrow \varphi^{C,x_C}} S(\varphi_{C,x_C}^{X_{(\cdot,j)} Z_{-j}} \| \psi_{C,x_C}^{X_{(\cdot,j)}} \otimes \sigma_{C,x_C}^{Z_{-j}}) \leq 2T + \log1/\lambda.$$ Using Quantum Raz's Lemma, we get
$$
	 \Ex_{C,x_C \leftarrow \varphi^{C,x_C}} \Ex_{i \in [n]} I(X_{(i,j)} : Z_{-j})_{\varphi_{C,x_C}} \leq 2(\log 1/\lambda + 2T)/n.
$$
By Markov's inequality, we have that there exists a $C, x_C, i \notin C$ such that
\begin{enumerate}
	\item Measuring the $X_{(i,\cdot)} A_{(i,\cdot)}$ register of $\varphi_{C,x_C}$ yields a tuple $(x_{(i,\cdot)},a_{(i,\cdot)})$ that satisfies $V(x_{(i,\cdot)},a_{(i,\cdot)}) = 1$ with probability at least $1 - \eps/8$.
	\item $S(\varphi_{C,x_C}^{X_{(i,\cdot)}} \| \mu)  \leq 32 t/n$.
	\item For all $j \in [k]$, $I(X_{(i,j)} : Z_{-j})_{\varphi_{C,x_C}} \leq 64 k(t + 2T)/n$.
\end{enumerate}
Let $\delta = 64 k(t + 2T)/n$. Let $s$ be the maximum number of qubits output by any one player in game $G$, so $ks \geq \log |\A|$. Then the total communication is $T = O(ks (t + \log k^2/\eps)/\sqrt{\eps})$. Then, if $t \leq c\epsilon^{3/2} n/ (k^4 s)$ for some universal constant $c$, we have $\delta \leq \eps/32k^2$. Let $\varphi = \varphi_{C,x_C}$. This yields the state and coordinate $i$ required.
\end{proof}

\subsection{The search protocol}

Next, we detail the search protocol used to construct $\ket{\psi}$ and $\ket{\varphi}$. We describe the protocol for a two-player game $G$; the extension to $k$ parties is straightforward.

Let $G = (\X \times \Y, \A \times \B, \mu, V)$ be a two-player free game, where $\X$ and $\Y$ are Alice and Bob's input alphabets, respectively, and $\A$ and $\B$ are their output alphabets. Consider the optimal strategy for $G^{\otimes n}$, where there is a shared state $\ket{\xi}^{E_A E_B}$ where on input $(x,y) \in \X^n \times \Y^n$, Alice and Bob apply measurements $\{M^x_a\}_{a \in \A^n}$ and $\{N^y_b\}_{b \in \B^n}$ respectively on their share of $\ket{\xi}$.

At the start of the search protocol, a multiset $C = \{i_1,\ldots,i_h\}$, $x_C \in \X^C$, and $y_C\in \Y^C$ are publically visible to Alice and Bob. They are both given the state
$$
	\ket{\psi^0_{C,x_C,y_C}} = \sum_{x \in \X^n,y\in\Y^n} \sqrt{\mu^{\otimes n}(x,y | x_C,y_C)} \ket{xxyy}^{XX'YY'} \ket{\xi}^{E_A E_B} \ket{0}^R
$$
where $\mu^{\otimes n}(x,y|x_C,y_C)$ is the distribution of $(x,y)$ conditioned on $x_C$, $y_C$.  Alice has access to registers $X X' E_A R$, and Bob has access to registers $E_B Y Y'$.

Then, Alice and Bob apply their measurements from the optimal strategy, controlled on the $X$ and $Y$ registers, respectively, to obtain
$$
	\ket{\psi^1_{C,x_C,y_C}} = \sum_{x \in \X^n,y\in\Y^n} \sqrt{\mu^{\otimes n}(x,y | x_C,y_C)} \ket{xxyy}^{XX'YY'} \sum_{a \in \A^n, b\in \B^n} \ket{\xi_{xyab}}^{E_A E_B} \ket{ab}^{AB} \ket{0}^R
$$
where $\ket{\xi_{xyab}} = (\sqrt{M^x_a} \otimes \sqrt{N^y_b}) \ket{\xi}$.

Alice and Bob then run a distributed search protocol controlled on the $XYAB$ registers. Fix $(x,y,a,b)$.  The protocol proceeds as follows: Alice and Bob divide the multiset $C$ into groups $D_1,\ldots,D_{q}$, each of size $m = \lceil 1/\eps' \rceil$. For each $\ell = 1,\ldots,q$,  Alice and Bob perform a distributed version of the Aaronson-Ambainis 3-dimensional search algorithm~\cite{aaronson2003quantum} to determine whether there is a coordinate $D_\ell$ contains a losing coordinate -- i.e., a coordinate $i \in D_\ell$ such that $V(x_i,y_i,a_i,b_i) = 0$. 

The search protocol for a group $D_\ell$ works as follows. Whenever the Aaronson-Ambainis algorithm is in the state $\sum_i \gamma_{i,z} \ket{i,z}$, where $\ket{i}$ corresponds to an index in $D_\ell$, and $\ket{z}$ is a qubit indicating whether a marked item has been found, the joint state between Alice and Bob will be $\sum_{i} \gamma_{i,z} \ket{i} \otimes \ket{z} \otimes \ket{i}$, where Alice holds the first $\ket{i}$ and $\ket{z}$, and Bob holds the second $\ket{i}$. Thus, Alice and Bob query locations are ``synchronized''. When Aaronson-Ambainis algorithm has to perform a query controlled on $\ket{i}$, Bob sends the qubit containing $\ket{b_i}$. Alice, controlled on $\ket{b_i}$, performs $\ket{z} \mapsto \ket{z \oplus V(x_i,y_i,a_i,b_i) \oplus 1}$ -- note that Alice can perform this, because in addition to $x_i$, $a_i$, and $b_i$, she also has access to $y_i$ because $y_C$ is public. We perform an additional XOR with $1$ because a ``marked item'' for the search algorithm corresponds to a \emph{losing} coordinate. Alice then sends back $\ket{b_i}$ to Bob. The other non-query transformations of the Aaronson-Ambainis algorithm are handled as in the the protocol described in~\cite{aaronson2003quantum}. Each step of the algorithm incurs at most $O(\log |\B|)$ qubits of communication, and there are $O(\sqrt{m})$ steps, resulting in $O(\sqrt{m} \log |\B|)$ qubits of total communication. If $D_\ell$ contains a losing coordinate, then this protocol will succeed in finding one with probability at least $2/3$.

If for at least one $\ell$, Alice and Bob find a losing coordinate for $G_\ell$, Alice sets the $R$ register to $0$; otherwise, it sets it to $1$.  Thus the total amount of communication of this protocol is $T = O(q\sqrt{m} \log |\B|) = O(\sqrt{1/\eps'} \log 1/\eta \log |\B|)$. The final state of the protocol looks like
$$
	\ket{\psi_{C,x_C,y_C}} = \sum_{x,y} \sqrt{\mu^{\otimes n}(x,y | x_C,y_C)} \ket{xxyy}^{XX'Y'Y} \sum_{a, b} \ket{\xi'_{xyab}}^{E'_A E'_B} \ket{ab}^{AB} (\alpha_{Cxyab} \ket{1}^R + \beta_{Cxyab} \ket{0}^R),
$$
where $\ket{\xi'_{xyab}} = \ket{\xi_{xyab}} \otimes \ket{w_{Cxyab}}$ with $\ket{w_{Cxyab}}$ denoting the workspace qubits of the two players.

Fix a setting of the registers $XYAB = (x,y,a,b)$. Suppose there was no $i \in [n]$ such that $V(x_i,y_i,a_i,b_i) = 0$. Then the search algorithm will never find a losing coordinate in any of the $G_\ell$'s, so for all $C$, the we have $\beta_{Cxyab} = 0$. On the other hand, suppose there were at least $\epsilon' n$ losing coordinates. We analyze, for a fixed $(x,y,a,b)$, the error quantity $\sum_C p(C) \, |\alpha_{Cxyab}|^2$. We can write $p(C) = \prod_\ell p(D_\ell)$, because each index in $C$ is chosen uniformly and independently at random. Furthermore, we can decompose $|\alpha_{Cxyab}|^2 = \prod_\ell |\alpha_{D_\ell xyab}|^2$, where $\alpha_{D_\ell xyab}$ is the probability amplitude that the Aaronson-Ambainis protocol does not find a losing coordinate in $D_\ell$. Thus the error quantity can be written as $\prod_\ell \sum_{D_\ell} p(D_\ell) |\alpha_{D_\ell xyab}|^2 = (\sum_D p(D) |\alpha_{Dxyab}|^2 )^q$. Each $D_\ell$ independently has at least  $1 - (1 - \eps')^{m} \geq 1 - 1/e$ probability of containing a losing coordinate. When $D_\ell$ has a losing coordinate, the Aaronson-Ambainis search protocol will succeed in finding it with probability at least $2/3$. Thus the error quantity is at most
\begin{align*}
	&\big ( \Pr(\text{$D$ contains losing coordinate}) \cdot (1/3) + \Pr(\text{$D$ does not have losing coordinate}) \cdot (1) \big)^q \\ 
	&\leq ( 1/3 + 1/e)^q \\
	&= \exp(-\Omega(q))= \eta.
\end{align*}
This establishes the requisite properties of the search protocol in the case of $k = 2$.

The extension to general $k$ parties is straightforward. At the beginning of the protocol, a multiset $C = D_1 \cdots D_q$ and inputs $x_C$ are publically visible to all players. They start with an analogous initial state $\ket{\psi_{C,x_C}}$, where each player $j$ has access to registers $X_{(\cdot,j)} X'_{(\cdot,j)} E_{j} A_{(\cdot,j)}$; player $1$ also has access to register $R$. They perform the distributed Aaronson-Ambainis protocol independently on all $D_\ell$. There are $k-1$ communication channels, one between the first player and all the other players. Whenever a query is to be made, player $j\in \{2,\ldots,k\}$ sends her answer symbol $a_{(i,j)}$ the first player, who then computes $V(x_{(i,\cdot)},a_{(i,\cdot)})$. The other non-query transformations of the algorithm are also easily extended to the multiplayer case. The total communication is $T = O(q\sqrt{m} \log |\A|) = O(\sqrt{1/\eps'} \log 1/\eta \log |\A|)$, where $\A$ is the output alphabet for all $k$ players.

\medskip
\noindent \textbf{Proof of Claim~\ref{clm:relative_min_entropy}}. Fix a $C, x_C$. Fix a player $j \in [k]$. Take the start state $\psi^0_{C,x_C}$ defined above (extended appropriately to $k$ players), and trace out the $X'_{(\overline{C},j)}$ register: $\theta^0_{C,x_C} = \tr_{X'_{(\overline{C},j)}} (\psi^0_{C,x_C})$. Since $\mu^{\otimes n}$ is a product distribution across players and also across game coordinates, we have that $\theta^0_{C,x_C} = U^{X_{(\overline{C},j)}} \otimes \phi^0_{C,x_C}$ where $U^{X_{(\overline{C},j)}}$ is the maximally mixed state for the register $X_{(\overline{C},j)}$, and
$$
\ket{\phi^0_{C,x_C}} = \ket{x_C}^{X_CX_C'}  \sum_{x_{(\cdot,-j)}} \sqrt{\mu_{-j}^{\otimes n}(x_{(\cdot,-j)}|x_C)} \ket{x_{(\cdot,-j)}x_{(\cdot,-j)}}^{X_{(\overline{C},-j)}X'_{(\overline{C},-j)}} \ket{\xi}^{E} \ket{0}^R $$
where $\mu_{-j}^{\otimes n}$ denotes the marginal distribution of $\mu^{\otimes n}$ on all players inputs, except for the $j$th player. Here, we used the simplifying assumption that $\mu$ is the uniform distribution. The search protocol described above never interacts with the $X'_{\overline{C}}$ registers. Thus, we can view the protocol as the $j$th player receiving a uniformly random input drawn from, $U^{X_{(\overline{C},j)}}$, and shares an entangled state $\phi^0_{C,x_C}$ with players $[k] - \{j\}$. The rest of the protocol is some two-way communication between player $j$ and every one else. 

We now wish to analyze the min-entropy of player $j$'s input register $X_{\overline{C},j}$ relative to the state of all other players. We appeal to the beautiful result of Nayak and Salzman~\cite{nayak2006limits}, whose theorem statement we reproduce here:

\begin{theorem}[\cite{nayak2006limits}]
\label{thm:nayak_salzman}
	Consider a communication protocol, without prior entanglement, where Alice receives a uniformly random $n$-bit input $X$, and interacts with Bob over a quantum communication channel. Let $\psi^{XB}$ be the final joint state of Alice's input $X$ and Bob's state in the protocol. Then, for any measurement strategy $\{M_x\}_x$ that Bob applies to his own state, the probability that Bob guesses Alice's input $X$ correct is at most $2^{2m_A}/2^n$, where $m_A$ is the number of qubits sent from Alice to Bob over the course of the protocol.
\end{theorem}
We now rephrase their theorem to use relative min-entropy instead of guessing probabilities. Let $\alpha$ be the optimal guessing probability for Bob. Then, the \emph{quantum conditional min-entropy} $\Hmin(X | B)_\psi$ is defined to be $-\log \alpha$. However, by SDP duality, we have the alternative characterization that $\Hmin(X|B)_\psi = -\inf_{\sigma^B} S_\infty (\psi^{XB} \| \I^X \otimes \sigma^B)$~\cite{konig2009operational}. Let $\sigma^B$ be a state achieving this infimum. Then $\log \alpha = S_\infty (\psi^{XB} \| \I^X \otimes \sigma^B) = S_\infty(\psi^{XB} \| \frac{1}{2^n} \I^X \otimes \sigma^B) - n$. By the theorem of Nayak and Salzman, $\log \alpha \leq 2m_A - n$, so $S_\infty(\psi^{XB} \| \frac{1}{2^n} \I^X \otimes \sigma^B) = S_\infty(\psi^{XB} \| \psi^X \otimes \sigma^B)  \leq 2m_A$, where we used the fact that $\psi^X$ is the uniform distribution.

To apply this theorem to our setting, we can treat player $j$ as ``Alice'' and the rest of the players as ``Bob''. Alice exchanges at most $T$ qubits with Bob. The crucial component of the Nayak-Salzman theorem is that Bob's probability of guessing does not depend on how many qubits he sent to Alice! Thus, we can imagine that in the beginning of the protocol he sent the $E_j$ register of the shared entangled state $\phi^0_{C,x_C}$ to Alice first. We have that there exists a state $\sigma_{C,x_C}^{Z_{-j}}$ such that
\begin{align*}
2T &\geq S_\infty (\psi_{C,x_C}^{X_{(\cdot,j)} Z_{-j}} \| \psi_{C,x_C}^{X_{(\cdot,j)}} \otimes \sigma_{C,x_C}^{Z_{-j}}).
\end{align*}

\section{Parallel repetition for free CQ games}
\label{sec:cq-proof}

\subsection{The model}
Here we introduce the model of $k$-player classical-quantum (CQ) games. A $k$-player classical-quantum (CQ) game $G$ is a tuple $(\X, A, \mu, \{V_x\}_{x \in \X})$, where 
	\begin{enumerate}
		\item $\X = \X_1 \times \X_2 \times \cdots \times \X_k$ with each $\X_j$ being a finite alphabet; 
		\item $A = A_1 \otimes A_2 \otimes \cdots \otimes A_k$, with each $A_j$ being a finite-dimensional complex Hilbert space;		
		\item $\mu$ is a probability distribution over $\X$; 
		\item For each $x \in \X$, $0 \preceq V_x \preceq \I$ is a positive semidefinite operator that acts on the space $A$.
	\end{enumerate}

In a $k$-player CQ game $G = (\X,A,\mu,\{V_x\})$, the referee will sample a tuple of inputs $x = (x_1,\ldots,x_k) \in \X$ from the distribution $\mu$, and send question $x_j$ to player $j$. Player $j$ will apply a local unitary on her part of a shared entangled state, and send her qubits in the space $A_j$ to the referee. The referee then performs the binary measurement $\{V_x, \I - V_x\}$ on the players' answers, and accepts if the outcome corresponding to $V_x$ is observed. In the case that the referee accepts, we say the players win the game $G$. 

A \emph{strategy} for a CQ game $G$ is a shared state $\ket{\xi}^{EA}$ (where $E$ and $A$ are $k$-partite spaces split between the $k$ players), and for each player $j$ a set of unitaries $\{U^j_{x_j}\}_{x_j \in \X_j}$, which act on the space $E_j A_j$. On input $x_j$, player $j$ applies the unitary $U^j_{x_j}$ to the $E_j A_j$ registers of $\ket{\xi}$, and then sends the $A_j$ register to the referee. The \emph{entangled value} of a CQ game $G$ is defined as the maximum probability a referee will accept over all possible (finite-dimensional) strategies for $k$ players:
$$
	\eval (G) = \max_{ \ket{\xi}, \{\{U^j_{x_j}\}_{x_j} \}_j} \Ex_{x \leftarrow \mu} \left[ \left\| \sqrt{V_x} U^1_{x_1} \otimes \cdots \otimes U^k_{x_k} \ket{\xi} \right \|^2 \right].
$$

The $n$-fold repetition of a CQ game $G = (\X,A,\mu,\{V_x\})$ is denoted by $G^{\otimes n} = (\X^n, B, \mu^{\otimes n}, \{W_{\vx}\}_{\vx \in \X^n} )$, where: $B$ is the tensor product of $n$ isomorphic copies of $A$; $\mu^{\otimes n}(\vx) = \prod_i \mu(\vx_i)$; and $W_{\vx} = \bigotimes_{i \in [n]} V^i_{\vx_i}$ with $V^i_{\vx_i}$ denoting the $V_{\vx_i}$ POVM element acting on the $i$th copy of $A$ in $B$.

The model of CQ games is a strict generalization of the standard notion of games with entangled players, where the inputs and outputs of the players are classical, and the verification predicate is some function of the inputs and outputs. 

\medskip
\noindent \textbf{Using quantum search with CQ-games?} One would hope that the analysis given in the previous section, which uses fast quantum search, would carry over to the setting of CQ games. However, there is an obstacle to this, which we do not know how to overcome: the first is that the analysis above (specifically, arguing property \textbf{(A)} in Lemma~\ref{lem:state_construction}) uses the fact that the $XA$ registers, in the standard basis, are either winning or not winning. In the CQ game case, we cannot say definitively whether the state in the $XA$ registers is a winning state or not, because the verifier measurement can be a general POVM. To perform the search protocol above, the players would have to measure their answer registers using the verifier measurement, but it is not clear whether the post-measurement state is actually \emph{useful} for the reduction in Theorem~\ref{thm:grover_pr}; in other words, property \textbf{(A)} may not necessarily hold. We leave it as an interesting open problem for whether one can prove an analogue of Theorem~\ref{thm:grover_pr} for CQ games.

Instead, will prove our parallel repetition theorem for CQ games using the approach of~\cite{JainPY14}, where we iteratively build a large collection of coordinates $i_1,\ldots,i_m$ such that $\Pr(\text{Win $i_k$} | \text{Win $i_1,\ldots,i_{k-1}$}) \leq 1 - \eps/2$ for all $k \leq m$. In the analysis, for each $k$, we only have to condition on winning coordinates $i_1,\ldots,i_{k-1}$, and thus we can leave the answer register for coordinate $i_k$ unmeasured, which gives us a useful advice state for the reduction.

\subsection{The proof}

%
Let $G = (\X,A,\mu,\{V_u\}_{u\in \X})$ be a $k$-player free CQ-game. Consider the repeated game $G^{\otimes n}$, and let $C \subseteq [n]$. Fix a strategy $\mathcal{S}$ for $G^{\otimes n}$; then $\Pr(\text{Win $C$})$ is the probability that, under $\mathcal{S}$, the answers of the $k$ players in coordinates indexed by $C$ pass the referee's verification procedure, over a random input $x$ drawn from $\mu^{\otimes n}$. 

\begin{lemma}
	Let $G = (\X,A,\mu,\{V_u\}_{u\in \X})$ be a $k$-player free CQ-game such that $\eval(G) = 1 - \eps$, and let $n \geq 1$. Let $s = \max_j \log \dim(A_j)$. Fix a strategy for the repeated game $G^{\otimes n}$. There exists a constant $c$ such that for all $C \subseteq [n]$ such that $|C| \leq c \eps n/s$, either $\Pr(\text{Win $C$}) < 2^{-c \eps n/k^2}$, or there exists an $i \notin C$ such that
	$$
		\Pr (\text{Win $i$} \big | \text{Win $C$}) \leq 1 - \eps/2.
	$$
\end{lemma}
\begin{proof} 
Again, we utilize Claim~\ref{clm:uniform_simulation} (which also holds for CQ games) to assume without loss of generality that the input distribution to $G$ is uniform. Consider a strategy for game $G^{\otimes n}$ where the shared state between the $k$ players is $\ket{\xi}^{EA}$ (where $A$ is a $n\times k$-partite register, and $E$ is a $k$-partite register) and each player $j \in [k]$ possesses a set of unitaries $\{W^j_{x_{(\cdot,j)}} \}_{x_{(\cdot,j)}}$ where, on input $x_{(\cdot,j)} \in \X_j^n$, player $j$ applies $W^j_{x_{(\cdot,j)}}$ to the $E_j A_{(\cdot,j)}$ registers of the shared state $\ket{\xi}$. For all $x \in \X^n$, let $\ket{\xi_x}^{EA} = \left( \bigotimes_j W^j_{x_{(\cdot,j)}} \right) \ket{\xi}^{EA}$.

Let $C \subseteq [n]$. Let
$$
	\ket{\psi}^{XX'EAR} := \sum_{x \in \X^n} \sqrt{\mu^{\otimes n}(x)} \ket{xx}^{XX'} \otimes \sum_{r \in \{0,1\}^C} \sqrt{V^r_{x_C}} \ket{\xi_x} ^{EA} \otimes \ket{r}^R,
$$
where $X$, $X'$ are $n\times k$-partite registers, and $V^r_{x_C}$ denotes $\bigotimes_{i \in C: r_i = 1} V^i_{x_{(i,\cdot)}} \bigotimes_{i \in C: r_i = 0} (\I - V^i_{x_{(i,\cdot)}})$, where by $V^i_u$ for $u \in \X$, we mean the POVM element $V_u$ acting on the registers $A_{(i,\cdot)}$. 

Let $\lambda = \Pr(\text{Win $C$})$, the probability of obtaining outcome $1^{|C|}$ when measuring the $R$ register of $\psi$ in the standard basis -- call this event $\mathcal{W}$. Let
$$
	\ket{\varphi}^{XX'EAR} := \frac{1}{\sqrt{\lambda}} \sum_{x \in \X^n} \sqrt{\mu^{\otimes n}(x)} \ket{xx}^{XX'} \otimes \sqrt{V_{x_C}} \ket{\xi_x} ^{EA} \otimes \ket{1}^R,
$$
where $V_{x_C} = \bigotimes_{i \in C} V^i_{x_{(i,\cdot)}}$.

Let $S \subseteq [n]$. When we write a state such as $\varphi_{x_S}$ (or $\psi_{x_S}$), we mean the pure state $\ket{\varphi}$ (or $\ket{\psi}$) conditioned on $X_S$ (and $X'_S$) register being equal to $x_S$.

If $\lambda < 2^{-c \eps n/k^2}$, we are done. Otherwise, assume that $\lambda \geq 2^{-c \eps n/k^2}$. Let $\overline{C} = [n] - C$. Note that $\lambda \varphi^{X'_{\overline{C}}XEA} + (1 - \lambda) \theta = \psi^{X'_{\overline{C}}XEA}$ for some state $\theta$. Then,
	\begin{align} 
		\log 1/\lambda &\geq S_\infty (\varphi^{X'_{\overline{C}}XEA} \| \psi^{X'_{\overline{C}}XEA}) \qquad &\text{(Fact~\ref{fact:max_divergence})} \nonumber \\
					&\geq S (\varphi^{X'_{\overline{C}}XEA} \| \psi^{X'_{\overline{C}}XEA}) \nonumber \\
					&\geq \Ex_{x_C \leftarrow \varphi^{X_C}} \left [S (\varphi_{x_C}^{X'_{\overline{C}}XEA} \| \psi_{x_C}^{X'_{\overline{C}}XEA}) \right] \qquad &\text{(Fact~\ref{fact:divergence_chain_rule})} \label{eq:div_bound}.
	\end{align}		

We first show that the distribution of individual coordinates $i \notin C$, conditioned on inputs in $C$ and the event $\mathcal{W}$, is only affected by a small amount (on average). Starting from line~\eqref{eq:div_bound},
\begin{align}
			\log 1/\lambda	&\geq \Ex_{x_C \leftarrow \varphi^{X_C}} \left [ S(\varphi_{x_C}^{X} \| \psi_{x_C}^X) \right] \qquad &\text{(Fact~\ref{fact:divergence_contractivity})} \nonumber \\
				&\geq \sum_{i \notin C}\Ex_{x_C \leftarrow \varphi^{X_C}} \left [  S (\varphi_{x_C}^{X_{(i,\cdot)}} \| \psi^{X_{(i,\cdot)}}) \right], \qquad &\text{(Fact~\ref{fact:divergence_split_rule})} \label{eq:x_i_div}
	\end{align}		
where we used the fact that $\psi_{x_C}^{X_{(i,\cdot)}} = \psi^{X_{(i,\cdot)}}$. 

Now we argue that, for every player $j$ and for most coordinates $i \notin C$, conditioning on inputs in $C$ and the event $\mathcal{W}$ does not introduce much correlation between $X_{(i,j)}$ (i.e. player $j$'s input for the $i$th coordinate) and the other players' quantum states and inputs.  
\begin{claim}
\label{clm:relative_min_entropy2} For all $j\in [k]$, and for all $x_C$, there exists a state $\sigma_{x_C}^{Z_{-j}}$ such that
	$$S_\infty (\psi_{x_C}^{X_{(\overline{C},j)} Z_{-j}} \| \psi^{X_{(\overline{C},j)}} \otimes \sigma_{x_C}^{Z_{-j}}) \leq 2 \log \dim(A_{(C,j)}),$$
where $Z_{-j} = X'_{(\overline{C},-j)} X_{(\cdot,-j)} E_{-j} A_{(\cdot,-j)}$. 
\end{claim}
We defer the proof of this claim for later. For now, we assume it. Then, using Fact~\ref{fact:relative_min_entropy_chain_rule} with line~\eqref{eq:div_bound} and Claim~\ref{clm:relative_min_entropy2}: 
\begin{align*}
\log 1/\lambda + 2\log \dim(A_{(C,j)}) \geq \Ex_{x_C \leftarrow \varphi^{X_C}} S (\varphi_{x_C}^{X_{(\overline{C},j)}  Z_{-j}} \| \psi^{X_{(\overline{C},j)}} \otimes \sigma_{x_C}^{Z_{-j}}).
\end{align*}
The states $\varphi_{x_C}^{X_{(\overline{C},j)}Z_{-j}}$ and $\psi^{X_{(\overline{C},j)}} \otimes \sigma_{x_C}^{Z_{-j}}$ satisfy the conditions of Quantum Raz's Lemma, and we get 
\begin{align}
	2 \left( \log 1/\lambda + 2\log \dim(A_{(C,j)}) \right) \geq \Ex_{x_C \leftarrow \varphi^{X_C}} \left[  \sum_{i \notin C} I(X_{(i,j)} : Z_{-j})_{\varphi_{x_C}} \right]. \label{eq:mutual_inf}
\end{align}
Lines~\eqref{eq:x_i_div} and~\eqref{eq:mutual_inf} yield:
\begin{enumerate}[label=(\roman*)]
	\item\label{item:x_i_div} $\Ex_{i \in \overline{C}} \Ex_{x_C \leftarrow \varphi^{X_C}} \left [ S (\varphi_{x_C}^{X_{(i,\cdot)}} \| \psi^{X_{(i,\cdot)}}) \right ] \leq (\log 1/\lambda)/|\overline{C}|$,
	\item\label{item:mutual_inf} For all $j \in [k]$, $\Ex_{i \in \overline{C}} \Ex_{x_C \leftarrow \varphi^{X_C}} \left[  I(X_{(i,j)} : Z_{-j})_{\varphi_{x_C}} \right] \leq 2 \left( \log 1/\lambda + 2\log \dim(A_{(C,j)}) \right)/|\overline{C}|$.
\end{enumerate}
In the above, the index $i$ is chosen uniformly at random from $\overline{C}$. Let $s = \max_j \log \dim(A_j)$, and let $\delta = 2( \log 1/\lambda + 2|C| s)/|\overline{C}|$. For each setting of $x_C$ and $i$ we use Lemma~\ref{lem:superjpy} on the pure state $\varphi_{x_C}$ to obtain for each player $j$ a set of unitaries $\{U^j_{i,x_C, u} \}_{u \in \X_j}$ such that
\begin{align*}
		&\Ex_{i \in \overline{C}} \Ex_{x_C \leftarrow \varphi^{X_C}}\Ex_{x_{(i,\cdot)} \leftarrow \varphi_{x_C}^{X_{(i,\cdot)}}}   \left[ K \left( \varphi_{x_C x_{(i,\cdot)}}  ,\mathcal{U}_{i,x_C, x_{(i,\cdot)}} (\varphi_{x_C}) \right) \right] \\ &\leq \Ex_{i \in \overline{C}} \Ex_{x_C \leftarrow \varphi^{X_C}} \left[ 4 k \sum_j I(X_{(i,j)} : Z_{-j})_{\varphi_{x_C}} \right ] \\
		&\leq 4k^2 \delta.
\end{align*}
where we let $U_{i,x_C,x_{(i,\cdot)}} =  \bigotimes_j U^j_{i,x_C,x_{(i,j)}} $, and let $\mathcal{U}_{i,x_C, x_{(i,\cdot)}}$ be the CP map that maps $\varphi \mapsto U_{i,x_C,x_{(i,\cdot)}} \varphi U_{i,x_C,x_{(i,\cdot)}}^\dagger$.

We now describe a protocol for the $k$ players to play game $G$. The players receive $u \in \X$, drawn from the product distribution $\mu$. Player $j$ receives $u_j \in \X_j$. For each value of $x_C$, the players share the state $\varphi_{x_C}^{XX'_{\overline{C}}EA}$. Note that these states are still pure, because we have conditioned on specific settings of $X_C$. Player $j$ has access to the $X_{(\cdot,j)}X_{(\cdot,j)}' E_j A_{(\cdot,j)}$ part of each state. The players also have access to common shared randomness.



\begin{figure}[H]
\begin{center}
\textbf{Protocol B} \\
\medskip
\framebox{
\begin{minipage}{0.9\textwidth}
	\textbf{Input}: $u \in \X$. Player $j$ receives $u_j$. \\
	\textbf{Preshared entanglement}: $\{\varphi_{x_C}^{XX'_{\overline{C}}EA}\}_{x_C}$ \\
\textbf{Strategy for player $j$}:
\begin{enumerate}
	\item Use shared randomness to pick an index $i \in \overline{C}$ uniformly at random.
	\item Use shared randomness to sample an $x_C \leftarrow \varphi^{X_C}$.
	\item Apply the local unitary $U^j_{i,x_C, u_j}$ on the $X_{(\cdot,j)} X'_{(\cdot,j)} E_j A_{(\cdot,j)}$ registers of $\varphi_{x_C}$.
	\item Output the $A_{(i,j)}$ part of $\varphi_{x_C}$. 
\end{enumerate}

\end{minipage}
}

\end{center}
\end{figure}

We now relate the winning probability of this protocol with the quantity $\omega = \Ex_{i \notin C} \Pr \left( \text{Win $i$}  \big | \text{Win $C$} \right)$. First, observe that for $i \notin C$:
\begin{align*}
\Pr \left( \text{Win $i$}  \big | \text{Win $C$} \right) &= \frac{1}{\lambda} \Pr \left( \text{Win $C \cup \{i\}$} \right) \\
&= \frac{1}{\lambda} \Ex_{x \leftarrow \mu^{\otimes n}} \left \| \sqrt{V_{x_{(i,\cdot)}}^i} \sqrt{V_{x_C}} \ket{\xi_x} \right \|^2 \\
&= \left \| \sqrt{ V_{x_{(i,\cdot)}}^i} \ket{\varphi} \right \|^2 \\
&= 	\Ex_{x_{(i,\cdot)} \leftarrow \varphi^{X_{(i,\cdot)}}} \left \| \sqrt{V_{x_{(i,\cdot)}}^i}  \ket{\varphi_{x_{(i,\cdot)}}} \right \|^2 .
\end{align*}
Let $\kappa$ denote the winning probability of Protocol B. This is equal to
\begin{align*}
	\kappa &=  \Ex_{u \leftarrow \mu} \Ex_{i} \Ex_{x_C \leftarrow \varphi^{X_C}} \left \| \sqrt{V_u^i}   U_{i,x_C, u} \, \ket{\varphi_{x_C}} \right \|^2\\
		&=  \Ex_{i} \Ex_{x_C \leftarrow \varphi^{X_C}} \Ex_{u \leftarrow \mu}  \left \| \sqrt{V_u^i}   U_{i,x_C, u} \, \ket{\varphi_{x_C}} \right \|^2 \\
		 &\geq \Ex_i \Ex_{x_C \leftarrow \varphi^{X_C}} \Ex_{u \leftarrow \varphi_{x_C}^{X_{(i,\cdot)}}} \left \| \sqrt{V_u^i}   U_{i,x_C, u} \, \ket{\varphi_{x_C}} \right \|^2 - 4\delta.
\end{align*}
where we use line~\eqref{eq:x_i_div} and appeal to Lemma~\ref{lem:brrrs-divergence}. Let $$\tau := \Ex_i \Ex_{x_C \leftarrow \varphi^{X_C}} \Ex_{u \leftarrow \varphi_{x_C}^{X_{(i,\cdot)}}} \left \| \sqrt{V_u^i}   U_{i,x_C, u} \, \ket{\varphi_{x_C}} \right \|^2.$$

For every $i \in [n]$, $u \in \X$, define the quantum operation $\E_{i,u}$ that, given a state $\varphi$, measures the $A_{(i,\cdot)}$ registers using $V^i_{u}$ measurement, and outputs a classical binary random variable $F$ indicating the verification measurement outcome (outcome $1$ corresponds to ``accept'' and outcome $0$ corresponds to ``reject''). Let 
$$F_0 =  \Ex_i \Ex_{x_{(i,\cdot)} \leftarrow \varphi^{X_{(i,\cdot)}}} \E_{i,x_{(i,\cdot)}} \left (  \varphi_{x_{(i,\cdot)}} \right ) =  \Ex_i  \Ex_{x_C \leftarrow \varphi^{X_C}} \Ex_{x_{(i,\cdot)} \leftarrow \varphi_{x_C}^{X_{(i,\cdot)}}} \E_{i,x_{(i,\cdot)}} \left ( \varphi_{x_C, x_{(i,\cdot)}}  \right ),$$
and let 
$$F_1 = 
\Ex_i \Ex_{x_C \leftarrow \varphi^{X_C}}  \Ex_{u \leftarrow \varphi_{x_C}^{X_{(i,\cdot)}}} \E_{i,x_{(i,\cdot)}} \left ( \mathcal{U}_{i,x_C, u} ( \varphi_{x_C})  \right).$$ 
Note that $\Pr(F_0 = 1) = \omega$, and $\Pr(F_1 = 1) = \tau$.

By the convexity of the squared Bures metric (Fact \ref{fact:fidelity_contractivity}) and Lemma~\ref{lem:superjpy},
\begin{align*}
	K(F_0,F_1) &\leq  \Ex_i \Ex_{x_C} \Ex_{x_{(i,\cdot)} \leftarrow \varphi_{x_C}^{X_{(i,\cdot)}}} \left[ K \left(\E_{i,x_{(i,\cdot)}} \left ( \varphi_{x_C, x_{(i,\cdot)}} \right ),\E_{i,x_{(i,\cdot)}} \left ( \mathcal{U}_{i,x_C, x_{(i,\cdot)}} ( \varphi_{x_C}) \right) \right)\right] \\
	&\leq \Ex_i \Ex_{x_C} \Ex_{x_{(i,\cdot)} \leftarrow \varphi_{x_C}^{X_{(i,\cdot)}}}  \left[ K \left( \varphi_{x_C, x_{(i,\cdot)}} ,\mathcal{U}_{i,x_C, x_{(i,\cdot)}} ( \varphi_{x_C}  \right)\right]  \\
	&\leq 4k^2 \delta .
\end{align*}

By assumption, $\log 1/\lambda < c \eps n/k^2$, and $|\overline{C}| \geq (1 - c\eps/s)n$. Thus, 
\begin{align*}
\delta &\leq \frac{2\gamma n}{|\overline{C}|} + 4\frac{(n - |\overline{C}|)s}{|\overline{C}|} \\
	  &\leq 4c \eps /k^2 + 8c \eps.
\end{align*}
 By choosing a small enough constant $c$, we can get that $\delta \leq \eps/24$ and $K(F_0,F_1) \leq \eps/6$. If $\omega \leq 1 - \eps/6$, we are done. Otherwise, by Lemma~\ref{lem:brrrs-fidelity}, $\Pr(F_1 = 1) \geq \omega - \eps/6$.

This means that $\kappa \geq \omega - \eps/6 - 4\delta \geq \omega - \eps/2$. On the other hand we have $\kappa \leq \eval(G)$, so thus $\omega \leq \eval(G) + \eps/2$ -- and hence by averaging there exists an $i \notin C$ such that  $\Pr \left( \text{Win $i$}  \big | \text{Win $C$} \right) \leq \eval(G) + \eps/2$. This concludes the proof.
\end{proof}

\begin{theorem}
	$\eval(G^{\otimes n}) \leq (1 - \eps^2)^{\Omega(n/sk^2)}$.
\end{theorem}

\begin{proof}
	For any subset $C \subseteq [n]$, we have that $\eval(G^{\otimes n}) \leq \Pr(\text{Win $C$})$. We construct a $C$ iteratively as follows. As long as there exists an index $i \notin C$ such that $\Pr (\text{Win $i$} \big | \text{Win $C$}) \leq 1 - \eps/2$, we add it to $C$. If at some point $\Pr (\text{Win $C$}) \leq 2^{-c \eps n/k^2}$, we are done, because this is at most $(1 - \eps)^{\Omega(n/k^2)}$. Otherwise, if we cannot find any such $i$'s to add, then it must be that $|C| > c \eps n/s$. But then, by Bayes' rule, we have that $\Pr (\text{Win $C$}) \leq (1 - \eps/2)^{|C|} \leq (1 - \eps^2)^{\Omega(n/s)}$. 
	
	In either case, $\eval(G^{\otimes n}) \leq (1 - \eps^2)^{\Omega(n/sk^2)}$.
\end{proof}

\medskip
\noindent \textbf{Proof of Claim~\ref{clm:relative_min_entropy2}}. Fix a player $j$. We need to present a low-cost communication protocol between player $j$ and all other players to produce the state $\ket{\psi_{x_C}}$. We will call player $j$ ``Alice'' and the rest of the players collectively as ``Bob''. 

The input $x_C$ to coordinates in $C$ is publically visible to both Alice and Bob. Alice and Bob start with the initial state
$$
	\ket{\psi^0_{x_C}} = \sum_{x} \sqrt{\mu^{\otimes n}(x | x_C)} \ket{xx}^{XX'} \otimes \ket{\xi}^{EA} \otimes \ket{0}^R,
$$
where Alice has registers $X_{(\cdot,j)} X_{(\cdot,j)}' E_j A_{(\cdot,j)} R$, and Bob has all other registers. Alice and Bob then locally apply the strategy $\mathcal{S}$ for $G^{\otimes n}$ that we used above, where Alice applies player $j$'s strategy, and Bob applies everybody else's. The state then becomes
$$
	\ket{\psi^1_{x_C}} = \sum_{x} \sqrt{\mu^{\otimes n}(x | x_C)} \ket{xx}^{XX'} \otimes \ket{\xi_x}^{EA} \otimes \ket{0}^R.
$$
Bob sends the answer qubits $A_{(C,-j)}$ to Alice, who performs the verification measurement $V_{x_C}$ on the $A_{(C,\cdot)}$ registers, and stores the outcomes in the $R$ register. Finally, Alice sends back registers $A_{(C,-j)}$ back to Bob. The final state of the protocol is precisely $\ket{\psi_{x_C}}$.

As in Claim~\ref{clm:relative_min_entropy}, this protocol never interacts with the $X'$ register. Thus, we can view the protocol as Alice receiving a uniformly random input $X_{(\overline{C},j)}$, Bob sending part of some shared state $\phi_{x_C}$ to Alice, and carrying out the rest of the protocol above. 

The total amount of communication sent from Alice to Bob in this protocol is the number of qubits in $A_{(C,-j)}$. Using Theorem~\ref{thm:nayak_salzman}, we get that there exists a state $\sigma^{Z_{-j}}_{x_C}$ such that 
	$$S_\infty (\psi_{x_C}^{X_{(\overline{C},j)} Z_{-j}} \| \psi^{X_{(\overline{C},j)}} \otimes \sigma_{x_C}^{Z_{-j}}) \leq 2 \log \dim(A_{(C,\cdot)}).$$

\section{A lower bound}
\label{sec:lower_bound}

In this section we demonstrate a $k$-player free CQ game $G$ such that $\val(G) = \eval(G) = 1/2$, but $\eval(G^{\otimes n}) = \val(G^{\otimes n}) = 1/2$ for $n \leq k$. This implies that in general, the dependence of the exponent of parallel repetition on the number of players is necessary (for both the classical and entangled value). This is a generalization of Feige's ``non-interactive agreement'' example of a $2$-player game $F$ such that $\val(F^{\otimes 2}) = \val(F)$.

\begin{theorem}
	There exists a $k$-player free CQ game $G$ and $n > 1$ such that $\eval(G^{\otimes n}) = \val(G^{\otimes n}) \geq \eval(G)^{n/k} = \val(G)^{n/k}$.
\end{theorem}
\begin{proof}
Consider the $k$-player game $G$ where each player $j$ receives uniformly random bit $x_j$ independent of the others. Each player $j$ has to output a pair $(i_j,a_j) \in [k] \times \{0,1\}$. The players win iff there exists an $i$ such that $i = i_1 = \cdots = i_k$, and $x_i = a^{\oplus}_{-i}$, the parity of the bits in the set $\{a_1,\ldots,a_k\} - \{a_i\}$.

To prove this theorem it suffices to show that the non-signaling value of $G$, $\val_{ns}(G)$, is $1/2$, and the classical value of the game, $\val(G^{\otimes n})$, is still $1/2$. This is because, for any game $F$, $\val_{ns}(F) \geq \eval(F) \geq \val(F)$~\cite{buhrman2013parallel}. 

It is clear that $\val_{ns}(G) \geq \eval(G) \geq \val(G) \geq 1/2$; a classical strategy achieving this is for every player to deterministically output $(1,0)$. We now show that $\val_{ns}(G) \leq 1/2$. Let $\mathcal{A} = [k] \times \{0,1\}$. A non-signaling strategy for $G$ is a conditional probability distribution $p((i_1,a_1),\ldots,(i_k,a_k) | x_1,\ldots,x_k)$, that satisfies the following conditions: for all subsets $I \subseteq [k]$, the complement $J = [k] - I$, 
$$
	\sum_{\alpha_J \in \mathcal{A}^J} p(\alpha_I, \alpha_J | x_I, x_J) = \sum_{\alpha_J \in \mathcal{A}^J} p(\alpha_I, \alpha_J | x_I, x'_J) \qquad \text{for all $x_J, x'_J \in \{0,1\}^J$ and $\alpha_I \in \mathcal{A}^I$},
$$
where for a set $S$,  $\alpha_S$ is a set of tuples $(i,a)$ indexed by elements in $S$, and $x_S$ indicates a set of bits indexed by elements of $S$. 

The probability that strategy $p$ wins game $G$ is
\begin{align*}
	&\Ex_{x_1,\ldots,x_k} \left[ \sum_{i,i'} \sum_{\substack{a_1,\ldots,a_k\\ a^{\oplus}_{-i} = x_i}} p((i,a_1),\ldots,(i',a_i),\ldots,(i,a_k)|x_1,\ldots,x_k) \right] \\
	&=\Ex_{x_1,\ldots,x_k} \left[  \Ex_{y} \sum_{i} \sum_{\substack{a_1,\ldots,a_k\\ a^{\oplus}_{-i} = y}} \sum_{i',b_i}  p((i,a_1),\ldots,(i',b_i),\ldots,(i,a_k)|x_1,\ldots,y,\ldots,x_k) \right]
\end{align*}
Applying the non-signaling constraints to the sum $ \sum_{i} \sum_{\substack{a_1,\ldots,a_k\\ a^{\oplus}_{-i} = y}} \sum_{i',b_i} p(\cdots)$ for when $y = 0$, then we get
\begin{align*}
 & \sum_{i} \sum_{\substack{a_1,\ldots,a_k\\ a^{\oplus}_{-i} = 0}} \sum_{i',b_i}  p((i,a_1),\ldots,(i',b_i),\ldots,(i,a_k)|x_1,\ldots,0,\ldots,x_k) \\  
  &= \sum_{i} \sum_{\substack{a_1,\ldots,a_k\\ a^{\oplus}_{-i} = 0}} \sum_{i',b_i}  p((i,a_1),\ldots,(i',b_i),\ldots,(i,a_k)|x_1,\ldots,1,\ldots,x_k)  \\
 &\leq 1 - \sum_{i} \sum_{\substack{a_1,\ldots,a_k\\ a^{\oplus}_{-i} = 1}} \sum_{i',b_i}  p((i,a_1),\ldots,(i',b_i),\ldots,(i,a_k)|x_1,\ldots,1,\ldots,x_k).
\end{align*}
But this implies that the non-signaling game value is at most $1/2$.

Now consider the repeated game $G^{\otimes k}$. We now give a strategy for the players such that $\val(G^{\otimes k}) = 1/2$. Therefore $\eval(G^{\otimes k})$ (and $\val_{ns}(G^{\otimes k})$) is at least $1/2$.

In the repeated game, each player $j$ receives a uniformly random vector of inputs $(x^1_j,x^2_j,\ldots,x^k_j)$. For the $\ell$'th repetition, player $j$ will output the pair $(\ell,x^j_j)$. The probability that the players win the first coordinate is $1/2$. Conditioned on the first coordinate winning, we have that $x^1_1 \oplus x^2_2 \oplus \cdots \oplus x^k_k = 0$. But then this ensures that the players win the rest of the coordinates with certainty.
\end{proof}

\section{Open problems}

We conclude with a variety of open problems.
\begin{enumerate}
	\item Is it possible to extend the Grover search analysis to handle CQ games? 
	
	\item Is \emph{strong parallel repetition} possible with the entangled value of free games? In other words, can the base of $1 - \eps^{3/2}$ of Theorem~\ref{thm:grover_pr_informal} be improved to $1 - \eps$? 
	
	\item Is the base of $1 - \eps^{2}$ for the repeated \emph{classical} value of free games tight? If so, this would mean that there is a separation of classical and quantum parallel repetition for free games.
	
	\item It was shown by~\cite{feige2002error} that the dependence on the output alphabet size, for classical parallel repetition, is necessary -- even for free games. However, Holenstein showed the repeated game value for non-signaling games has no such alphabet dependence~\cite{holenstein2007parallel}. Is this dependence necessary for the quantum case?

	\item Parallel repetition is but one way to amplify hardness of games. XOR lemmas and direct product threshold theorems give other ways of amplifying hardness. Can one use this communication protocol perspective to give unified proofs of these hardness amplification results?
	
	
	\item Can we identify an interesting class of games for which we can prove improved parallel repetition theorems, by designing efficient communication protocols to generate advice states?
		
	\item The mantra, ``Better parallel repetition theorems from better communication protocols,'' suggests an intriguing connection between games and communication protocols. Although games are protocols that forbid communication between the players, one can define the \emph{communication complexity of a game} as the minimum communication needed for the players to determine whether they have won or lost the game. Our mantra suggests a relationship between the value and communication complexity of a game. What is the nature of this relationship? 
	
	\item Can one use these techniques to prove parallel repetition for entangled games with an arbitrary input distribution?

\end{enumerate}

\medskip
\vspace{2em}
\noindent \textbf{Acknowledgments}. We thank Scott Aaronson for helpful comments about the Aaronson-Ambainis algorithm. We also thank Thomas Vidick and Andre Chailloux for helpful comments on earlier versions of this paper. HY is supported by an NSF Graduate Fellowship Grant No. 1122374 and National Science Foundation Grant No. 1218547. XW is funded by ARO contract W911NF-12-1-0486 and by the NSF Waterman Award of Scott Aaronson. Part of this research was conducted while XW was a Research Fellow, and HY a visiting graduate student, at the Simons Institute for the Theory of Computing, University of California, Berkeley.

\bibliographystyle{alphaabbrvprelim}
\bibliography{bibliography}


\appendix 

\section{$k$-player parallel repetition for classical games}
\label{sec:classical_pr}


Let $G = (\X,\A,\mu,V)$ be a $k$-player free game. Consider the repeated game $G^{\otimes n}$, and consider a classical strategy for the repeated game. Without loss of generality, a strategy for the $k$ players is set of functions $\{h_j: \X_j^n \to \A_j^n\}$, whereupon input $x_j \in \X_j^n$, player $j$ outputs $h_j(x_j)$. We note that this proof follows that of~\cite{barak2009strong} very closely. While this result is considered folklore, we include it for the sake of completeness. 

\begin{lemma}
	Let $G = (\X,\A,\mu,V)$ be a $k$-player free game such that $\val(G) = 1 - \eps$, and let $n \geq 1$. Let $s$ denote the maximum length of a single player's output. Fix a classical strategy $\{h_j\}$ for the repeated game $G^{\otimes n}$. There exists a constant $c > 0$ such that for all $C \subseteq [n]$ such that $|C| \leq c \eps n/sk$, either $\Pr(\text{Win $C$}) < 2^{-c \eps n}$, or there exists an $i \notin C$ such that
	$$
		\Pr (\text{Win $i$} \big | \text{Win $C$}) \leq 1 - \eps/2.
	$$
\end{lemma}
\begin{proof}
	Let $C\subseteq [n]$. Let $V:\X \times \A \to \{0,1\}$ denote the acceptance predicate used by the referee in the single game $G$. Define the following probability distribution on $\X_{\overline{C}} \times \X_C \times \A_C \times \{0,1\}$:
	$$
		\psi(x_{\overline{C}},x_C,a_C,r) := \left \{ 
		\begin{array}{ll}
		\Pr(x_{\overline{C}}) \Pr(x_C) & a_C = h(x)|_C, \text{and } r = V(x_C,a_C) \\
		0									& \text{otherwise}. 
		\end{array}
		\right.
	$$
	where $\Pr(x_{\overline{C}})$ and $\Pr(x_C)$ is the probability of the input $x_{\overline{C}}$ and $x_C$ respectively, under the distribution $\mu$; and $h(x)|_C := ( h_1(x_1)|_C,\ldots,h_k(x_k)|_C )$, where $h_j(x_j)|_C$ denotes the restriction of player $j$'s output to the coordinates in $C$. Slightly abusing notation, by $V(x_C,a_C)$ we mean $\prod_{i \in C} V(x_{(i,\cdot)},a_{(i,\cdot)})$.
	Now define the conditioned distribution
	$$
		\varphi(x_{\overline{C}},x_C,a_C) = \psi(x_{\overline{C}},x_C,a_C| r = 1),
	$$
	or in other words, the distribution of $\psi$ conditioned on $r = 1$ (i.e. the predicate $V$ accepting). Call this event $\mathcal{W}$. In what follows, we will use $X_{\overline{C}}$, $X_C$, $A_C$ to denote the random variables whose joint probability distribution is given by either $\varphi$ or $\psi$, depending on the context. 
	
	Let $\lambda = \Pr(\text{Win $C$}) = \Pr_\psi(r = 1)$. If $\lambda < 2^{-c \eps n}$, we are done. Assume otherwise. By Fact~\ref{fact:max_divergence}, we have $\log 1/\lambda \geq S_\infty (\varphi \| \psi)$. In what follows, we use $\varphi^{X_C}$ to denote the marginal distribution of $\varphi$ on $x_C$. Similarly, $\varphi^{X_{(i,\cdot)}}$ and $\psi^{X_{(i,\cdot)}}$ are the marginals on the appropriate variables.

We now argue the existence of a probability distribution $\sigma^{X_C A_C}$ such that
	$$S_\infty (\psi^{X_{\overline{C}} X_C A_C} \| \psi^{X_{\overline{C}}} \otimes \sigma^{X_C A_C}) \leq 2 |C| sk$$
where $s$ is the maximum length of a single player's output.

Consider the following two-party communication protocol: Alice receives $X_{\overline{C}}$ as input. Bob possesses the random variables $X_C$. Bob creates a copy of this random variable; call it $X'_C$. He sends Alice $X'_C$. Alice then applies the strategy $\{h_j\}$ on the inputs $X_{\overline{C}} X_C$ to produce outputs $A_C$. She sends Bob the random variable $A_C$. The final state of the protocol, where we have ``traced out'' (i.e. marginalized over) $X'_C$, is precisely $\psi^{X_{\overline{C}} X_C A_C}$. Using Theorem~\ref{thm:nayak_salzman}, and the argumentation used in the proof of Claim~\ref{clm:relative_min_entropy}, we obtain such a distribution $\sigma^{X_C A_C}$.\, \footnote{Technically, the arguments in Claim~\ref{clm:relative_min_entropy} were for quantum min-entropy. However, this gives only a stronger result in the classical setting.}

Using Fact~\ref{fact:relative_min_entropy_chain_rule2}, we get that $S(\varphi^{X_{\overline{C}} X_C A_C} \|  \psi^{X_{\overline{C}}} \otimes \sigma^{X_C A_C}) \leq \log 1/\lambda + 2 |C| sk$. Using Facts~\ref{fact:divergence_chain_rule} and~\ref{fact:divergence_split_rule}, we get that
	\begin{align}
			\Ex_{i \notin C} \Ex_{(x_C,a_C) \leftarrow \varphi^{X_C A_C}} S (\varphi_{x_C,a_C}^{X_{(i,\cdot)}} \| \psi^{X_{(i,\cdot)}}) \leq (\log 1/\lambda + 2|C| sk)/|\overline{C}|,\label{eq:x_i_div_classical}
	\end{align}
where $\varphi^{X_{(i,\cdot)}}_{x_C,a_C}$ is the marginal distribution of the random variable $X_{(i,\cdot)}$ under $\varphi$ conditioned on $X_C = x_C$, and $A_C = a_C$. 
Let $\delta = 2 \left( \log 1/\lambda + 2|C| sk \right)/|\overline{C}|$. 

We now describe a protocol for the $k$ players to play game $G$. The players receive $u \in \X$, drawn from the product distribution $\mu$. Player $j$ receives $u_j \in \X_j$. The players have access to shared randomness. In the protocol below, if a player $j$ decides to abort, it outputs a random answer symbol in $\A_j$.

\begin{figure}[H]
\begin{center}
\textbf{Protocol C} \\
\medskip
\framebox{
\begin{minipage}{0.9\textwidth}
	\textbf{Input}: $u \in \X$. Player $j$ receives $u_j$. \\
\textbf{Strategy for player $j$}:
\begin{enumerate}
	\item Use shared randomness to pick an index $i \in \overline{C}$ uniformly at random.
	\item Use shared randomness to sample $(x_C,a_C) \leftarrow \varphi(x_C,a_C)$.
	\item If $\varphi(x_{(i,j)} = u_j | x_C,a_C) = 0$, then abort. Otherwise, continue.
	\item Sample $x_{(\cdot,j)} \leftarrow \varphi(x_{(\cdot,j)} | x_C,a_C,x_{(i,j)} = u_j)$.
	\item Output the $i$th coordinate of $h_j(x_{(\cdot,j)})$.
\end{enumerate}

\end{minipage}
}

\end{center}
\end{figure}

\begin{claim}
	Fix $x_C,a_C$ in the support of $\varphi^{X_CA_C}$. For all $x_{(\cdot,1)} \cdots x_{(\cdot,k)}$,
	$$
		\varphi(x_{(\cdot,1)} \cdots x_{(\cdot,k)} | x_C,a_C) = \prod_j \varphi(x_{(\cdot,j)} | x_C,a_C).
	$$
\end{claim}
\begin{proof}
	Fix $x_C, a_C$ in the support of $\varphi^{X_CA_C}$. Then, for all $x_{(\cdot,1)} \cdots x_{(\cdot,k)}$,
	$$
		\varphi(x_{(\cdot,1)} \cdots x_{(\cdot,k)} | x_C,a_C) = \prod_j \varphi(x_{(\cdot,j)} | x_{(\cdot,<j)}, x_C, a_C),
	$$
	where $x_{(\cdot,<j)} = x_{(\cdot,1)} \cdots x_{(\cdot,j-1)}$. Suppose for now that the left hand-side is non-zero. Thus, every factor on the right hand side must also be non-zero. Fix a $j$. Since $(x_C,a_C)$ determines the event $\mathcal{W}$, we have
	\begin{align*}
		\varphi(x_{(\cdot,j)} | x_{(\cdot,<j)}, x_C,a_C) &= \psi(x_{(\cdot,j)} | x_{(\cdot,< j)}, x_C,a_C) \\
		&= \frac{\psi(x_{(\cdot,< j)} | x_{(\cdot,j)}, x_C, a_C) \psi(x_{(\cdot, j)} | x_C, a_C)}{\psi(x_{(\cdot,< j)} | x_C, a_C) }.
	\end{align*}
	Since for all $j'$, $a_{(C,j')}$ is determined by $x_{(\cdot,j')}$, and under $\psi$ the inputs to the players are all independent, we have that $\psi(x_{(\cdot,< j)} | x_{(\cdot,j)}, x_C, a_C) = \psi(x_{(\cdot,< j)} | x_C,a_{(C,< j)}) = \psi(x_{(\cdot,< j)} | x_C, a_C)$. This implies that $\varphi(x_{(\cdot,j)} | x_{(\cdot,<j)}, x_C,a_C) = \psi(x_{(\cdot,j)} | x_C,a_C) =  \varphi(x_{(\cdot,j)} | x_C,a_C)$. This proves the claim in the case that $\varphi(x_{(\cdot,1)} \cdots x_{(\cdot,k)} | x_C,a_C) > 0$.

	


When $\varphi(x_{(\cdot,1)} \cdots x_{(\cdot,k)} | x_C,a_C) = 0$, let $j$ be the least such $j$ with $\varphi(x_{(\cdot,j)} | x_{(\cdot,<j)}, x_C, a_C) = 0$. If $j = 1$, then we are done. If $\Pr(x_{(\cdot,j)}| x_C, a_C) = 0$, we are also done. Otherwise, $\psi(x_{(\cdot,<j)} | x_C, a_C) > 0$ and $\psi(x_{(\cdot,j)}| x_C, a_C) > 0$, and this implies that $\psi(x_{(\cdot,\leq j)} | x_C, a_C) > 0$ (because the players' inputs are all independent), a contradiction. This concludes the proof.
\end{proof}

\begin{corollary}
	\label{clm:conditional_independence}
	Fix $x_C,a_C$ in the support of $\varphi^{X_CA_C}$. Let $i\notin C$ and $u \in \X$ be such that $\varphi(x_{(i,\cdot)} = u | x_C,a_C) > 0$. Then for all $x_{(\cdot,1)} \cdots x_{(\cdot,k)}$,
	$$
		\varphi(x_{(\cdot,1)} \cdots x_{(\cdot,k)} | x_C,a_C,x_{(i,\cdot)} = u) = \prod_j \varphi(x_{(\cdot,j)} | x_C,a_C,x_{(i,j)} = u_j).
	$$
\end{corollary}
\begin{proof}
	Note that 
\begin{align*}
\varphi(x | x_C,a_C,x_{(i,\cdot)} = u) &= \frac{\varphi(x \wedge x_{(i,\cdot)} = u | x_C,a_C)}{\varphi(x_{(i,\cdot)} = u | x_C,a_C)} \\
&= \frac{\prod_j  \varphi(x_{(\cdot,j)} \wedge x_{(i,j)} = u_j | x_C,a_C)}{\varphi(x_{(i,\cdot)} = u | x_C,a_C)}.
\end{align*}
The proof of the previous Claim shows $\varphi(x_{(i,\cdot)} = u | x_C,a_C) = \prod_j \varphi(x_{(i,j)} = u_j | x_C,a_C)$. This completes the argument.
\end{proof}

We now relate the winning probability of protocol C with the quantity $\omega = \Ex_{i \notin C} \Pr \left( \text{Win $i$}  \big | \text{Win $C$} \right)$. For $i \notin C$: 
\begin{align*}
	\Pr \left( \text{Win $i$}  \big | \text{Win $C$} \right) &= \Ex_{x \leftarrow \varphi^{X}} V(x_{(i,\cdot)},h(x)_i) \\
	&= \Ex_{x_C,a_C} \left [ \Ex_{x_{(i,\cdot)} | x_C,a_C}\left[ \Ex_{x | x_C,a_C,x_{(i,\cdot)}} V(x_{(i,\cdot)},h(x)_i) \right] \right]
\end{align*}
where we view the sampling of $x \leftarrow \varphi^{X}$ as a three stage process:
\begin{enumerate}
	\item $(x_C,a_C) \leftarrow \varphi(x_C,a_C)$
	\item $x_{(i,\cdot)} \leftarrow \varphi(x_{(i,\cdot)} | x_C,a_C)$
	\item $x \leftarrow \varphi(x | x_C,a_C,x_{(i,\cdot)})$
\end{enumerate}

\noindent Let $\kappa$ denote the winning probability of Protocol C. Fix a $u \in \X$, $i \in \overline{C}$, and $x_C, a_C$. Let $P_{u,i,x_C,a_C}$ denote the probability that the players win the game, given that their input is $u$, they sampled $i$ and $(x_C,a_C)$ in the protocol. Suppose that $\varphi(x_{(i,\cdot)} = u,x_C,a_C) > 0$, then the players do not abort, and so the $x \in \X^n$ sampled by all the players in Protocol C in this case is distributed as $\prod_j \varphi(x_{(\cdot,j)} | x_C,a_C,x_{(i,j)} = u_j)$. Then, by Corollary~\ref{clm:conditional_independence}, this is equal to $\varphi(x | x_C,a_C, x_{(i,\cdot)} = u)$. Thus we have the following:
\begin{align*}
	\kappa &\geq \Ex_{u \leftarrow \mu}  \Ex_{i} \Ex_{x_C,a_C} P_{u,i,x_C,a_C} \\
	&\geq \left( \Ex_{i} \Ex_{x_C,a_C} \left [\Ex_{x_{(i,\cdot)} | x_C,a_C}   P_{x_{(i,\cdot)},i,x_C,a_C} \right] \right )  - 4\delta\\
	 &=\left(  \Ex_{i}  \Ex_{x_C,a_C}\left [ \Ex_{x_{(i,\cdot)} | x_C,a_C} \left [ \Ex_{x|x_C,a_C,x_{(i,\cdot)}} V(x_{(i,\cdot)},h(x)_i) \right]\right] \right) - 4\delta
\end{align*}
where $(x_C,a_C) \leftarrow \varphi^{X_CA_C}$, and $x_{(i,\cdot)}$ is distributed according to $\varphi^{X_{(i,\cdot)}}_{x_C,a_C}$. In the inequality we used line~\eqref{eq:x_i_div_classical} and Lemma~\ref{lem:brrrs-divergence}. This is precisely $\omega - 4\delta$. Since Protocol C is a valid strategy for game $G$, we have that $\Ex_{i \in \overline{C}} \Pr(\text{Win $i$} | \text{Win $C$}) \leq \val(G) + 4\delta$. By our assumption on $|C|$ and $\Pr(\text{Win $C$})$, we can choose an appropriate constant $c$ so that there exists an $i \notin C$ with $\Pr(\text{Win $i$} | \text{Win $C$}) \leq \val(G) + 4\delta \leq 1 - \eps/2$. 
\end{proof}

\begin{theorem}
	$\val(G^{\otimes n}) \leq (1 - \eps^2)^{\Omega(n/sk)}$.
\end{theorem}

\end{document}